\documentclass[11pt,fleqn]{article}
\usepackage{amssymb,latexsym,amsmath,amsfonts,graphicx}
\usepackage{pictex,color,comment}
\usepackage{pict2e}

    \advance\textheight by \topskip

    \numberwithin{equation}{section}

\def\beq{\begin{equation}}
\def\eeq{\end{equation}}
\def\bdf{\begin{definition}}
\def\edf{\end{definition}}

\def\bit{\begin{itemize}}
\def\eit{\end{itemize}}

\makeatletter
\def\eqalign#1{\null\vcenter{\def\\{\cr}\openup\jot\m@th
  \ialign{\strut$\displaystyle{##}$\hfil&$\displaystyle{{}##}$\hfil
      \crcr#1\crcr}}\,}
\makeatother

\newcommand{\bea}{\begin{align}}
\newcommand{\eea}{\end{align}}

\newcommand{\be}{\begin{equation}}
\newcommand{\ee}{\end{equation}}

    \def\e{{\epsilon}}

    \def\Re{{\rm Re \,}}
    \def\Im{{\rm Im \,}}

    \def\bigO{{\cal O}}

    \def\P2n{{\rm P}_{{\rm II}}^{(n)}}

    \newtheorem{theorem}{Theorem}[section]
    \newtheorem{lemma}[theorem]{Lemma}
    \newtheorem{corollary}[theorem]{Corollary}

    \newtheorem{Definition}[theorem]{Definition}
    \newenvironment{definition}{\begin{Definition}\rm}{\end{Definition}}
    \newtheorem{Remark}[theorem]{Remark}
    \newenvironment{remark}{\begin{Remark}\rm}{\end{Remark}}
    \newtheorem{Example}[theorem]{Example}
    
    \newtheorem{Assumptions}[theorem]{Assumptions}
    
    \newtheorem{Notation}[theorem]{Notation}

    \newenvironment{proof}%
    {\rm \trivlist \item[\hskip \labelsep{\bf Proof. }]}%
    {\hspace*{\fill}$\Box$\endtrivlist}
    {\rm \trivlist \item[\hskip \labelsep{\bf Proof}]}%
    {\hspace*{\fill}$\Box$\endtrivlist}

\begin{document}
\title{On the ratio probability of the smallest eigenvalues in the \\ Laguerre Unitary Ensemble}
\author{Max R. Atkin\footnote{Institut de Recherche en Math\'ematique et Physique,  Universit\'e
catholique de Louvain, Chemin du Cyclotron 2, B-1348
Louvain-La-Neuve, BELGIUM}, Christophe Charlier\footnotemark[\value{footnote}], Stefan Zohren\footnote{Quantum  and  Nanotechnology  Theory  Group,  Department  of  Materials,  and
Machine  Learning  Research  Group,  Department  of  Engineering  Science,  University  of  Oxford,  UK}}

\maketitle

\begin{abstract}
We study the probability distribution of the ratio between the second smallest and smallest eigenvalue in the $n\times n$ Laguerre Unitary Ensemble. The probability that this ratio is greater than $r>1$ is expressed in terms of an $n \times n$ Hankel determinant with a perturbed Laguerre weight. The limiting probability distribution for the ratio as $n\to\infty$ is found as an integral over $(0,\infty)$ containing two functions $q_{1}(x)$ and $q_{2}(x)$. These functions satisfy a system of two coupled Painlev\'{e} V equations, which are derived from a Lax pair of a Riemann-Hilbert problem. We compute asymptotic behaviours of these functions as $rx \to 0_{+}$ and $(r-1)x \to \infty$, as well large $n$ asymptotics for the associated Hankel determinants in several regimes of $r$ and $x$.
\end{abstract}

%


\section{Introduction and main results}

The Laguerre Unitary Ensemble (LUE) consists of the space of $n\times n$ complex positive definite Hermitian matrices endowed with the distribution
\begin{equation}\label{intro LUE}
\frac{1}{\widetilde{Z}_{n,\alpha}} (\det M)^{\alpha} e^{-n \rm{ Tr} M}  dM, \qquad \alpha > -1,
\end{equation}
where $\widetilde{Z}_{n,\alpha}$ is the normalisation constant and $dM$ is the Lebesgue measure
\begin{equation}
dM = \prod_{i=1}^{n} dM_{ii} \prod_{1 \leq i<j \leq n } d \mathrm{Re} M_{ij}d \mathrm{Im} M_{ij}.
\end{equation}
The probability measure \eqref{intro LUE} is invariant under unitary conjugation and induces a joint probability distribution on the eigenvalues $\lambda_{1},...,\lambda_{n}$ on $(\mathbb{R}^{+})^{n}$ given by 
\begin{equation} \label{eq:LUE}
\frac{1}{n! \widehat Z_{n,\alpha}} \Delta_n(\lambda)^2  \prod^n_{i=1}e^{-n \lambda_i} \lambda_i^\alpha \chi_{\mathbb{R}^+}(\lambda_i) d\lambda_i,
\end{equation}
where $\Delta_n(\lambda)$ denotes the Vandermonde determinant
\begin{equation}
\Delta_n(\lambda)\equiv  \Delta_n(\lambda_{1},...,\lambda_{n}) := \prod_{1 \leq i<j \leq n} (\lambda_j-\lambda_i)
\end{equation}
and $\chi_{\mathbb{R}^+}$ is the characteristic or indicator function with support on the positive half line.
The normalisation constant $\widehat Z_{n,\alpha}$, also known as the partition function, is given by
\begin{equation}\label{lol43}
\widehat Z_{n,\alpha} =  \frac{1}{n!} \int^{\infty}_{0} ... \int_{0}^{\infty} \Delta_n(\lambda)^2 \prod^{n}_{i=1}  e^{-n \lambda_i}  \lambda_i^\alpha d\lambda_i.
\end{equation}


The theory of random matrices has enjoyed a growing interest and study for a number of decades in part due to the surprising number of connections between it and seemingly unrelated topics, see for example \cite{KeaSna}. Some classical references in the fields are \cite{AkeBaiFra,AndGuiZei,Mehta,Tao}. Probably the earliest appearance of random matrix theory dates back to Wishart in 1928 in the context of multivariate data analysis \cite{Wishart}. He was studying matrices of the form $X^{*}X$, where $X$ is an $m \times n$ matrix ($m \geq n$) whose entries $X_{ij}$ are independent and identically distributed complex Gaussian variables
\begin{equation}
\Re X_{ij}, \Im X_{ij} \sim \mathcal{N}\left(0,\sigma^{2} = \frac{1}{2n}\right),
\end{equation}
and $X^{*}$ is the conjugate transpose of $X$. Such positive semi-definite matrices, called Wishart matrices, possess eigenvalues which are also distributed according to \eqref{eq:LUE}, where $\alpha = m-n$ is an integer. Since then, the LUE has been studied a lot, and has found applications in different areas, for example in finance (see e.g. \cite{BaiSil}, Section 12.2).

\subsection*{Some known results}

The limiting mean eigenvalue density function is given by the Marchenko-Pastur \cite{MarPas} law
\begin{equation}
d\mu(x) = \rho(x)dx = \frac{1}{2\pi} \sqrt{\frac{4-x}{x}} dx,
\end{equation}
supported on the interval $[0,4]$. Besides the eigenvalue density, other quantities of interest are related to the extreme value statistics of the eigenvalues. The most well-known extreme value statistics in the field of random matrix theory is the Tracy-Widom distribution \cite{TraWidGUE} which describes the properly rescaled fluctuations of the largest (or smallest) eigenvalue of a matrix from the Gaussian Unitary Ensemble (GUE). This distribution is given by $\widetilde F_{2}(x) = \exp \left(- \int_{x}^{\infty} (s-x)u^{2}(s)ds \right)$ and $u$ is the Hastings-McLeod \cite{HasMcL} solution of the Painlev\'e II equation $u^{\prime\prime}(s) = su(s) + 2u(s)^{3}$ satisfying the boundary condition $u(s) \sim \mathrm{Ai}(s)$ as $s \to \infty$, where $\rm{Ai}$ is the Airy function. For the case of the LUE, the distribution of the rescaled fluctuations of the largest eigenvalue at the soft edge, at $x=4$, is also given by the Tracy-Widom distribution $\widetilde F_{2}$. 

In this paper we focus on the hard edge, at $x=0$, where the distribution of the smallest eigenvalue (denoted by $\lambda_{\mathrm{min}}$) is given by another distribution $F_{\alpha}(x)$ in terms of a transcendent of the Painlev\'e V equation, also proved by Tracy and Widom \cite{TraWidLUE}. More precisely, one has
\begin{equation} \label{intro4 lambdamindist}
F_{\alpha} (x)=\lim_{n\to\infty} \mathbb{P}_{n,\alpha}(4 n^2 \lambda_{\mathrm{min}} > x  ) = \exp\left( - \frac{1}{4} \int_0^x \log\left(\frac{x}{\xi}\right) q^2(\xi) d\xi \right)
\end{equation}
where $q(x)$ is the solution of the equation
\begin{equation}\label{Painleve V lol}
xq\big(1-q^{2}\big)\big(xqq^{\prime}\big)^{\prime} + x\left( (xq^{\prime})^{\prime} + \frac{q}{4} \right)\big(1-q^{2}\big)^{2} + x^{2} q \left( qq^{\prime} \right)^{2} = \alpha^{2} \frac{q}{4},
\end{equation}
with boundary condition $q(\xi)\sim J_{\alpha}(\sqrt{\xi})$ for $\xi \to 0_{+}$, and $J_{\alpha}$ is the Bessel function of the first kind of order $\alpha$ (see \cite[Section 10.2 and and Section 10.7]{NIST} for definition and properties of this function). 

That $q(\xi)$ is indeed a transcendent of the Painlev\'e V equation can be seen from the following transformation
\begin{equation*}
q(\xi) = \frac{1+y(x)}{1-y(x)}, \quad \xi= x^2,
\end{equation*}
from which it can be readily checked that $y(x)$ is a solution to the Painlev\'e V equation with suitable chosen coefficients. 

Another quantity of interest is the gap probability
\begin{equation}
G_{n,\alpha} (d) = \mathbb{P}_{n,\alpha} \left( \lambda_{\mathrm{smin}} - \lambda_{\mathrm{min}} > d \right), \qquad d > 0,
\end{equation}
between the smallest and second smallest (denoted by $\lambda_{\mathrm{smin}}$) eigenvalue of the LUE. Some results were obtained in \cite{ForWit}, where it was shown that the density of $G_{\alpha}(x) = \lim_{n\to\infty} G_{n,\alpha} (\frac{x}{4n^{2}})$ exists and is characterised by the solution of a Painlev\'e III equation and its associated linear isomonodromic system. For results on the gap probability but at the soft edge, see \cite{PerSch} and \cite{WitBorFor}. In \cite{PerSch2}, using heuristics arguments and numerical simulations, the authors generalized the results obtained in \cite{PerSch} to a more general context.

The main focus of this paper is the distribution of the ratio 
\begin{equation}
Q_{n,\alpha}(r) = \mathbb{P}_{n,\alpha}  \left(\frac{\lambda_{\mathrm{smin}}}{\lambda_{\mathrm{min}}} > r\right), \qquad r > 1,
\end{equation}
between the second smallest and the smallest eigenvalue of the LUE. Note that the ratio distribution $Q_{n,\alpha}$ cannot be straightforwardly related to the gap $G_{n,\alpha}$ and the distribution of the smallest eigenvalue, since the three variables $\lambda_{\min}$, $\lambda_{\mathrm{smin}} - \lambda_{\mathrm{min}}$ and $\frac{\lambda_{\mathrm{smin}}}{\lambda_{\mathrm{min}}}$ are not independent. 

Our techniques differ from the techniques used in \cite{ForWit}, \cite{PerSch} and \cite{WitBorFor}. We express the ratio probability in terms of a Hankel determinant and then apply well-known rigorous techniques from Riemann-Hilbert (RH) problems analysis to derive asymptotics. We obtain a description of this quantity in terms of a solution $(q_{1},q_{2})$ to a system of two coupled Painlev\'{e} V equations arising from a Lax pair of a RH problem.


\subsection*{Statement of results}

We begin the calculation of the ratio probability between the second smallest and smallest eigenvalue for general $\alpha > - 1$ by writing the quantity $Q_{n,\alpha}(r)$ as an integral over a Hankel determinant. For $r > 1$, by definition of \eqref{eq:LUE} we have
\begin{equation*}
\label{P def}
Q_{n,\alpha} (r) =  n   \int_{0}^\infty e^{-ny}y^{\alpha} \left(   \int_{yr}^\infty ... \int_{yr}^\infty  \frac{1}{n! \widehat Z_{n,\alpha}} \Delta_{n-1}^2(\lambda)  \prod_{i=1}^{n-1}e^{-n \lambda_i} \lambda_i^\alpha (\lambda_{i}-y)^{2} d\lambda_i   \right) dy,
\end{equation*} 
where $y$ can be interpreted as the smallest eigenvalue. By changing variables $n\lambda_{i} = (n-1)\tilde\lambda_{i}$, we can then write 
\begin{equation} \label{Qn}
Q_{n,\alpha} (r) = \frac{1}{\widehat{Z}_{n,\alpha}} \left( \frac{n-1}{n} \right)^{(n-1)(n+1+\alpha)} \int_{0}^\infty y^\alpha e^{-n y} Z_{n-1,\alpha}\left(\frac{n}{n-1} y;r \right) dy,
\end{equation}
where we defined $Z_{n,\alpha}(y;r)$ by
\begin{equation}\label{def of Zn}
Z_{n,\alpha}(y;r) = \frac{1}{n!} \int_{yr}^\infty ... \int_{yr}^\infty \Delta_n(\lambda)^2 \prod_{i=1}^{n} (\lambda_i - y)^2 \lambda_i^\alpha e^{- n \lambda_{i}} d\lambda_i.
\end{equation}
Note that $Z_{n,\alpha}(y;r)$ is the Hankel determinant (\cite[equations (2.2.7) and (2.2.11)]{Szego OP}) with respect to the weight $w(x)$ defined by
\begin{equation} \label{weight}
w(x) = (x-y)^{2} x^{\alpha} e^{-n x} \chi_{[yr,\infty)}(x),
\end{equation}
where $\chi_{[yr,\infty)}(x)$ is the characteristic function of $[yr,\infty)$, i.e. we have
\begin{equation}
Z_{n,\alpha}(y;r) = \det \left( \int_{yr}^{\infty} x^{i+j}w(x)dx \right)_{i,j=0,...,n-1}.
\end{equation}

Our main results concern asymptotics for the Hankel determinants $Z_{n,\alpha}(y;r)$ and the limiting distribution as $n\to \infty$ of the ratio probability $Q_{n,\alpha}(r)$ which can be expressed in a compact form through a solution $(q_{1},q_{2})$ of a system of coupled Painlev\'{e} V equations. \vspace*{0.4cm}

\hspace{-0.6cm}\textbf{Asymptotics for the Hankel determinant $Z_{n,\alpha}(y;r)$}

\vspace{0.2cm}\hspace{-0.6cm}Let us define $s := 4n^{2}y$, which is a rescaling of $y$. We provide large $n$ asymptotics for $Z_{n,\alpha}(y;r)$ in three different regimes:
\begin{align}
& \mbox{Case 1: }(s,r) \mbox{ are in a compact subset of } (0,\infty)\times (1,\infty), \label{Case 1} \\
& \mbox{Case 2: }rs \to 0, \label{Case 2} \\
& \mbox{Case 3: }\frac{rs}{n} \to 0 \mbox{ and } (r-1)s \to \infty. \label{Case 3}
\end{align}
\begin{theorem}\label{P asymptotics} 
Let $\alpha >-1$ be fixed. As $n\rightarrow \infty$ and simultaneously $s$ and $r$ satisfy one of the three cases presented in \eqref{Case 1}, \eqref{Case 2} and \eqref{Case 3}, we have
\begin{equation}\label{asymptotics result for Hankel determinant}
\log Z_{n,\alpha}\left(\frac{s}{4n^{2}};r\right)  - \log (\widehat Z_{n,\alpha+2})  =   I(s;r) + \left\{ \begin{array}{l l}
\bigO(n^{-1}), & \mbox{ for Case }1, \\
\bigO(n^{-1}), & \mbox{ for Case }2, \\
\bigO\big( \tfrac{(rs)^{2}}{n} \big), & \mbox{ for Case }3,
\end{array} \right. 
\end{equation}
where 
\begin{equation}\label{def of I}
I(s;r) = \displaystyle -\frac{1}{4} \int_{0}^{s} (q_{1}^{2}(x;r)+rq_{2}^{2}(x;r))\log \left( \frac{s}{x}\right) dx.
\end{equation}
The functions $q_{1}^{2}(x;r)$ and $q_{2}^{2}(x;r)$ are real and analytic for $x \in (0,\infty)$ and $r \in (1,\infty)$, and they satisfy the following system of coupled Painlev\'{e} V equations:
\begin{align}
& \hspace{-0.7cm} xq_{1}\bigg(1-\sum_{j=1}^{2}q_{j}^{2}\bigg) \sum_{j=1}^{2}(xq_{j}q_{j}^{\prime})^{\prime} + \bigg[ x \left( (xq_{1}^{\prime})^{\prime} + \frac{q_{1}}{4}\right) + \frac{1}{q_{1}^{3}} \bigg]\bigg(1-\sum_{j=1}^{2}q_{j}^{2}\bigg)^{2} + x^{2}q_{1}\bigg( \sum_{j=1}^{2}q_{j}q_{j}^{\prime} \bigg)^{2} = \frac{\alpha^{2} q_{1}}{4}, \nonumber \\
& \hspace{-0.7cm} xq_{2}\bigg(1-\sum_{j=1}^{2}q_{j}^{2}\bigg) \sum_{j=1}^{2}(xq_{j}q_{j}^{\prime})^{\prime} + x \left( (xq_{2}^{\prime})^{\prime} + \frac{rq_{2}}{4}\right) \bigg(1-\sum_{j=1}^{2}q_{j}^{2}\bigg)^{2} + x^{2}q_{2}\bigg( \sum_{j=1}^{2}q_{j}q_{j}^{\prime} \bigg)^{2} = \frac{\alpha^{2} q_{2}}{4}, \label{system of q_1 q_2}
\end{align}
where primes denote derivatives with respect to  $x$. Furthermore, the functions $q_{1}$ and $q_{2}$ satisfy the following boundary conditions: as $(r-1)x \to \infty$, we have
\begin{align}
& q_{1}^{2}(x) = \frac{2}{\sqrt{(r-1)x}}+\bigO\left( \frac{1}{(r-1)x} \right), \label{asymp q1 inf} \\
& q_{2}^{2}(x) = 1 - \frac{\alpha}{\sqrt{rx}} - \frac{2}{\sqrt{(r-1)x}}+ \bigO\left(\frac{1}{(r-1)x}\right), \label{asymp q2 inf}
\end{align}
and as $rx \to 0$, we have
\begin{align}
& q_{1}(x) = \sqrt{\frac{2}{\alpha+2}}(1+\bigO(rx)), \label{asymp q1 origin} \\
& q_{2}(x) = (1-r^{-1})J_{\alpha+2}(\sqrt{rx})(1+\bigO(rx)) = \frac{1-r^{-1}}{2^{\alpha+2}\Gamma(\alpha+3)}\sqrt{rx}^{\alpha+2}(1+\bigO(rx)). \label{asymp q2 origin}
\end{align}
\end{theorem}

\begin{remark}
We will prove in the present paper that the system \eqref{system of q_1 q_2} with boundary conditions \eqref{asymp q1 inf}, \eqref{asymp q2 inf}, \eqref{asymp q1 origin} and \eqref{asymp q2 origin} possesses at least one solution ($q_{1}$,$q_{2}$), but there is no guaranty of uniqueness of this solution. Therefore, $q_{1}$ and $q_{2}$ are not defined through this system, but they are explicitly constructed from the solution of a model Riemann-Hilbert problem, whose solution (denoted $\Phi$) exists and is unique. This Riemann-Hilbert problem is presented in Section \ref{section RH model Phi}.
\end{remark}

\begin{remark}
In the regime as $(r-1)x \to \infty$ in \eqref{asymp q2 inf}, there are different cases. For example, note that if $x\to \infty$ and $r \to 1$, then the $\bigO\big(((r-1)x)^{-1}\big)$ term is larger than $(rx)^{-1/2}$ if $(r-1) \sqrt{x} \to 0$. In this case, \eqref{asymp q2 inf} can be rewritten as 
\begin{equation}
q_{2}^{2}(x) = 1 - \frac{2}{\sqrt{(r-1)x}}+ \bigO\left(\frac{1}{(r-1)x}\right), \qquad \mbox{ as } x \to \infty, r \to 1 \mbox{ and } (r-1) \sqrt{x} \to 0.
\end{equation}
\end{remark}
\begin{remark}
The system \eqref{system of q_1 q_2} are two coupled Painlev\'{e} V equations. It is worth to compare it with the Painlev\'{e} V equation given by \eqref{Painleve V lol}, and also to compare \eqref{def of I} with the Tracy-Widom distribution \eqref{intro4 lambdamindist}. This system is similar to the one obtained in \cite{CharlierDoeraene}, where the authors obtained a system of $k$ ($k \in \mathbb{N}_{0}$) coupled Painlev\'{e} V equations. The main difference here lie in the $q_{1}^{-3}$ extra term in the first equation of the system \eqref{system of q_1 q_2}, and in the small $rx$ asymptotics of $q_{1}$, which does not involve Bessel functions.
\end{remark}
\begin{corollary}\label{coro1}
As $rs \to 0$, we have
\begin{align}
& I(s;r) = \frac{-s}{2(\alpha+2)}(1+\bigO(rs)). \label{asymp I small}
\end{align}
As $s \to \infty$ and $r$ is in a compact subset of $(1,\infty)$, we have
\begin{align}
& I(s;r) = -\frac{rs}{4} + \alpha \sqrt{rs} + 2\sqrt{(r-1)s} + \bigO(\log s). \label{asymp I large}
\end{align}
\end{corollary}
\begin{proof}
To obtain \eqref{asymp I small}, it suffices to  substitute asymptotics \eqref{asymp q1 origin} and \eqref{asymp q2 origin} into \eqref{def of I}. To prove large $s$ asymptotics of $I(s;r)$ given by \eqref{asymp I large}, we decompose the integral into several parts as follows
\begin{equation}\label{several integrals in the proof}
\begin{array}{r c l}
\displaystyle I(s;r) & = & \displaystyle I_{1} + I_{2} + I_{3}, \\[0.1cm]
\displaystyle I_{1} & = & \displaystyle -\frac{1}{4} \int_{0}^{\frac{1}{rM}} (q_{1}^{2}(x;r)+rq_{2}^{2}(x;r))\log \left( \frac{s}{x}\right) dx, \\[0.3cm]
\displaystyle I_{2} & = & \displaystyle -\frac{1}{4} \int_{\frac{1}{rM}}^{\frac{M}{r}} (q_{1}^{2}(x;r)+rq_{2}^{2}(x;r))\log \left( \frac{s}{x}\right) dx, \\[0.44cm]
\displaystyle I_{3} & = & \displaystyle -\frac{1}{4} \int_{\frac{M}{r}}^{s} (q_{1}^{2}(x;r)+rq_{2}^{2}(x;r))\log \left( \frac{s}{x}\right) dx,
\end{array}
\end{equation}
where $M$ is a sufficiently large but fixed constant. Asymptotics \eqref{asymp q1 origin} and \eqref{asymp q2 origin} allow us to write $|I_{1}| = \bigO(\log s)$, as $s \to \infty$. For $I_{2}$, the parameters $(x,r)$ which appear in the functions $q_{1}$ and $q_{2}$ lie in a compact subset of $(0,\infty)\times (1,\infty)$, and thus we also have $|I_{2}| = \bigO(\log s)$ as $s \to \infty$. Over the domain of integration of $I_{3}$, the parameters $x$ and $r$ inside $q_{1}$ and $q_{2}$ satisfy $xr \geq M$, and thus we can use \eqref{asymp q1 inf} and \eqref{asymp q2 inf} to estimate it. We obtain
\begin{equation}
I_{3} = -\frac{rs}{4}+\alpha \sqrt{rs} + 2 \sqrt{(r-1)s} + \bigO(\log s), \qquad \mbox{ as } s \to \infty,
\end{equation}
which finishes the proof.
\end{proof}

\begin{remark}\label{Remark another paper}
Theorem \ref{P asymptotics} and Corollary \ref{coro1} provide asymptotics for $Z_{n,\alpha}\left( \frac{s}{4n^{2}};r\right)$ and $I(s;r)$ in various regimes of $n$, $s$ and $r$, which are useful to prove Theorem \ref{thm-Q} below for the limiting distribution of the ratio. Note that asymptotics \eqref{asymp I large} hold for $s \to \infty$ and $r$ in a compact subset of $(1,\infty)$, and not in the more general situation of $(r-1)x \to \infty$. The reason for that is, as it can be seen in the proof of Corollary \ref{coro1}, and more particularly in \eqref{several integrals in the proof}, to estimate $I_{2}$ we also need to find asymptotics for $q_{1}(x;r)$ and $q_{2}(x;r)$ in the two following cases: a) $xr = \bigO(1)$ and simultaneously $r \to 1$ and b) $xr = \bigO(1)$ and simultaneously $r \to \infty$. These asymptotics are also needed to obtain asymptotics as $r \to 1$ and as $r \to \infty$ for the limiting probability distribution of the ratio. These cases deserve another long and separate analysis and we intend to pursue this in another paper. We expect these asymptotics to be described in terms of a transcendental function, solution of a differential equation similar to the Painlev\'{e} V equation given by  \eqref{Painleve V lol}. 
\end{remark}

\textbf{Limiting probability distribution of the ratio $\frac{\lambda_{\mathrm{smin}}}{\lambda_{\min}}$}
\begin{theorem} \label{thm-Q} Let $\alpha>-1$ and $r >1$ be fixed. As $n\rightarrow \infty$, the limit $\displaystyle Q_\alpha (r):= \lim_{n\to\infty} Q_{n,\alpha} (r)$ exists and is given by
\begin{equation} \label{final expression for Q}
Q_\alpha (r)= \frac{1}{4^{\alpha+1}\Gamma(\alpha+1)\Gamma(\alpha+2)}\int_0^\infty  x^\alpha e^{I(x;r)} dx,
\end{equation}
where $I(x;r)$ given in Theorem \ref{P asymptotics}.
\end{theorem}

\begin{remark}
Let $p_{n}(x_{1},x_{2})$ denote the joint density for the first two smallest eigenvalues at $\lambda_{\min} = x_{1}$ and $\lambda_{\mathrm{smin}} = x_{2}$, which can be straightforwardly obtained by integrating \eqref{eq:LUE}, and is given by
\begin{equation*}
p_{n}(x_{1},x_{2}) = \frac{e^{-n(x_{1}+x_{2})}(x_{1}x_{2})^{\alpha}(x_{2}-x_{1})^{2}}{(n-2)! \widehat{Z}_{n,\alpha}}\int_{x_{2}}^{\infty} \hspace{-0.09cm}... \int_{x_{2}}^{\infty} \hspace{-0.07cm} \prod_{3 \leq i < j \leq n} \hspace{-0.08cm}(\lambda_{j}-\lambda_{i})^{2} \prod_{i=3}^{n}e^{-n\lambda_{i}}\lambda_{i}^{\alpha}(x_{1}-\lambda_{i})^{2}(x_{2}-\lambda_{i})^{2}d\lambda_{i}.
\end{equation*}
The ratio probability $Q_{n,\alpha}(r)$ is expressed in terms of $p_{n}$ by the relation
\begin{equation}\label{lol40}
Q_{n,\alpha}(r) = \int_{0}^{\infty}\int_{yr}^{\infty} p_{n}(y,y_{2})dy_{2}dy.
\end{equation}
In \cite{ForHug}, the authors expressed the density $p_{n}$ in the special case where $\alpha$ is an integer, as a determinant involving the Laguerre polynomials, they obtained \cite[formula (3.20)]{ForHug}
\begin{equation}\label{lol39}
p_{n}(x_{1},x_{2}) = n^{4}e^{-n(x_{1}+(n-1)x_{2})}\left( \frac{x_{2}}{x_{1}} \right)^{\alpha} (x_{2}-x_{1})^{2} D_{n}(x_{1},x_{2}),
\end{equation}
where
\begin{equation}
D_{n}(x_{1},x_{2}) = (-1)^{\frac{(\alpha+1)(\alpha+2)}{2}}\det \left[ \begin{array}{l}
\displaystyle \left[ \left. \partial_{t}^{(j+k-2)}L_{\alpha+n-1}^{(-\alpha+1)}(t) \right|_{t = -nx_{2}} \right]_{\substack{j=1,...,\alpha \\ k=1,...,\alpha+2}} \\
\displaystyle \left[ \left. \partial_{t}^{(j+k-2)}L_{\alpha+n-1}^{(-\alpha+1)}(t) \right|_{t = -n(x_{2}-x_{1})} \right]_{\substack{j=1,2 \\ k=1,...,\alpha+2}}
\end{array} \right],
\end{equation}
and $L_{j}^{(\alpha)}$ is the generalized Laguerre polynomial of degree $j$ and index $\alpha$. These polynomials are defined for $\alpha \in \mathbb{R}$ (not necessarily for $\alpha >-1$) through the recursive relations
\begin{equation*}
L_{0}^{(\alpha)}(x) = 1, \quad L_{1}^{(\alpha)}(x) = 1+\alpha-x, \mbox{ and } L_{k+1}^{(\alpha)}(x) = \frac{(2k+1+\alpha-x)L_{k}^{(\alpha)}(x)-(k+\alpha)L_{k-1}^{(\alpha)}(x)}{k+1}, \mbox{ } k \geq 1.
\end{equation*}
By combining \eqref{lol40} and \eqref{lol39}, this gives a determinantal representation for $Q_{n,\alpha}(r)$ if $\alpha \in \mathbb{N}$. Also, in \cite[equations (3.34) and (3.35)]{ForHug}, they obtain the following determinantal expression for the limiting density of the two smallest eigenvalues:
\begin{equation}\label{lol42}
p(s_{1},s_{2}) = \lim_{n\to\infty} \left( \frac{1}{4n^{2}} \right)^{2} p_{n}\left( \frac{s_{1}}{4n^{2}},\frac{s_{2}}{4n^{2}} \right) = \frac{e^{-\frac{s_{2}}{4}}}{16} \left( \frac{s_{2}}{s_{1}} \right)^{\alpha}D(s_{1},s_{2}),
\end{equation}
where
\begin{equation*}
D(s_{1},s_{2}) = \lim_{n\to\infty} \left( \frac{s_{2}-s_{1}}{4n^{2}} \right)^{2}D_{n} \left( \frac{s_{1}}{4n^{2}},\frac{s_{2}}{4n^{2}} \right) = \det \left[ \begin{array}{l}
\displaystyle \left[ I_{j-k+2}(\sqrt{s_{2}}) \right]_{\substack{j=1,...,\alpha \\ k=1,...,\alpha+2}} \\
\displaystyle \left[ \left( \frac{s_{2}-s_{1}}{s_{2}} \right)^{\frac{k-j}{2}} I_{j-k+2}(\sqrt{s_{2}-s_{1}}) \right]_{\substack{j=1,2 \\ k=1,...,\alpha+2}}
\end{array} \right]
\end{equation*}
From the change of variables $y = \frac{s}{4n^{2}}$ and $y_{2} = \frac{s_{2}}{4n^{2}}$ in \eqref{lol40} and then taking the limit $n \to \infty$, we have
\begin{equation}\label{lol41}
Q_{\alpha}(r) = \lim_{n\to\infty} \left( \frac{1}{4n^{2}} \right)^{2} \int_{0}^{\infty}\int_{rs}^{\infty} p_{n}\left( \frac{s}{4n^{2}},\frac{s_{2}}{4n^{2}}\right)ds_{2}ds = \int_{0}^{\infty} \int_{rs}^{\infty} p(s,s_{2})ds_{2}ds.
\end{equation}
The fact that the limit exists and can be interchanged with the integrals is not direct, and can be justified as in \cite[Proposition 5.11]{ForWit}. The formulas \eqref{lol42} and \eqref{lol41} give an explicit determinantal representation for $Q_{\alpha}(r)$ in terms of Bessel functions if $\alpha \in \mathbb{N}$.
\end{remark}

\subsection*{Outline}

In Section \ref{Section OP and diff identity}, we introduce a family of monic orthogonal polynomials in terms of which the Hankel determinant $Z_{n,\alpha}(y;r)$ can be expressed. We also use the RH problem for orthogonal polynomials introduced by Fokas, Its and Kitaev \cite{FokasItsKitaev} and derive a differential identity in $y$ for $Z_{n,\alpha}(y;r)$.

We apply the Deift/Zhou steepest descent method \cite{DeiftZhou1992,DKMVZ2,DKMVZ1} on this RH problem in Section \ref{Section Steepest descent Y} to obtain large $n$ asymptotics of $Z_{n,\alpha}(y;r)$ uniformly in $y$ small enough. 

In the analysis we will need a non standard model RH problem, which we introduce in Section \ref{section RH model Phi}. We derive a system of two coupled Painlev\'{e} V equations using a Lax pair in Section \ref{Section Lax Pair} and show asymptotic properties of certain solutions of these equations as $(r-1)x \to \infty$ and $rx\to 0$ in Section \ref{Section asymptotics of the special solutions}.

Finally we give a proof of Theorem \ref{P asymptotics} and Theorem \ref{thm-Q} in Section \ref{proof1} and \ref{proof2} respectively, by integrating the differential identity and using equation \eqref{Qn}.

\section{Differential identity for the Hankel determinant $Z_{n,\alpha}(y,r)$}\label{Section OP and diff identity}

In this section we relate the Hankel determinant $Z_{n,\alpha}(y;r)$ to a RH problem by making use of orthogonal polynomials. We consider a family of monic orthogonal polynomials $p_{j}$ of degree $j$ characterised by the relations
\begin{equation}\label{ortho p}
\int^\infty_{yr} p_j(x) p_m(x) w(x) dx = h_j \delta_{j m}, \qquad j,m = 0,1,2,...,
\end{equation}
where the weight $w$ is defined in \eqref{weight} and $h_{j}$ is the squared norm of $p_{j}$, which can be expressed in terms of Hankel determinants (see \cite[equations (2.1.5) and (2.1.6)]{Szego OP}) as follows:
\begin{equation} \label{squared norm}
h_{j} = \frac{Z_{j+1,\alpha}(y;r)}{Z_{j,\alpha}(y;r)}, \qquad Z_{0,\alpha}(y;r) := 1.
\end{equation}
It will prove useful for the later analysis to consider $Y(z) = Y_{n}(z;y,r)$ the matrix valued function defined by,
\begin{equation}
\label{Ysol}
Y(z) = \begin{pmatrix}p_n(z)&q_n(z)\\ \displaystyle -\tfrac{2\pi i}{h_{n-1}} p_{n-1}(z)& \displaystyle -\tfrac{2\pi i}{h_{n-1}} q_{n-1}(z)\end{pmatrix},
\end{equation}
where $q_{j}$ is the Cauchy transform of $p_{j}$ defined by
\begin{equation}
q_j(z) = \frac{1}{2\pi i} \int^{\infty}_{yr} \frac{p_j(x) w(x)}{x-z} dx.
\end{equation}
The function $Y$ can be characterised as the unique function satisfying a set of conditions \cite[equations (3.19)-(3.21)]{FokasItsKitaev}, known as the RH problem for $Y$, which are as follows:
\subsubsection*{RH problem for $Y$}
\begin{itemize}
\item[(a)] $Y: \mathbb{C}\setminus [yr, \infty) \rightarrow
  \mathbb{C}^{2 \times 2} $ is analytic.
\item[(b)] The limits of $Y(z)$ as $z$ approaches $(yr, \infty)$ from above and below exist, are continuous on $(yr, \infty)$ and are denoted by $Y_+$ and $Y_-$ respectively. Furthermore they are related by
\beq
Y_+(x) = Y_-(x) \begin{pmatrix}1&w(x)\\0&1\end{pmatrix},\qquad x\in (yr, \infty).
\eeq
\item[(c)] $Y(z) = (I+\bigO(z^{-1}))z^{n\sigma_3}$ as $z \rightarrow \infty$, where $\sigma_{3} = \begin{pmatrix}
1 & 0 \\ 0 & -1
\end{pmatrix}$.
\item[(d)] $Y(z) = Y_{yr}(z) \begin{pmatrix}
1 & -\frac{w(z)}{2\pi i }  \log(y r -z) \\ 0 & 1
\end{pmatrix}$ as $z \rightarrow yr$, where the principal branch of the logarithm is taken, and where $Y_{yr}(z)$ is analytic in a neighbourhood of $yr$.
\end{itemize}
\begin{remark}\label{remark: def of Res}
Since $Y$ is discontinuous on $(yr,\infty)$, the function $z  \mathrm{Tr} \left(Y^{-1}(z)Y^{\prime}(z) \sigma_3\right)$ is not analytic in a neighbourhood of $\infty$. Nevertheless, since $w(x)$ which appear in the jumps for $Y$ becomes exponentially small for large $x$, the non-analytic part of $z  \mathrm{Tr} \left(Y^{-1}(z)Y^{\prime}(z)  \sigma_3\right)$ in a neighbourhood of $\infty$ is also exponentially small in $z$. In fact, from \eqref{Ysol}, we have
\begin{equation}
Y^{\prime}(z) = \left( \frac{n}{z}\sigma_{3}+\bigO(z^{-2}) \right)z^{n\sigma_{3}}, \qquad \mbox{ as } z \to \infty,
\end{equation}
and thus
\begin{equation}\label{def of c_1}
z  \mathrm{Tr} \left(Y^{-1}(z)Y^{\prime}(z)  \sigma_3\right) = 2n + \frac{c_{1}}{z} + \bigO(z^{-2}), \qquad \mbox{ as } z \to \infty.
\end{equation}
for a certain $c_{1} \in \mathbb{C}$. This constant will play a role in Lemma \ref{diff identity} below.
\end{remark}

In \cite{XuZhao,XuZhao2}, the authors considered a similar but different RH problem, where they perturbed the classical Jacobi ensemble by adding $n$-dependent singularities to the weight. Having introduced the above objects we are now in a position to state the following lemma which is central for the asymptotic analysis of $Z_{n,\alpha}(y;r)$. The identity \eqref{equation diff} and its proof are similar to the one performed in \cite{AtkClaMez}.
\begin{lemma} \label{diff identity} The following identity holds,
\begin{equation} \label{equation diff}
\partial_y \log Z_{n,\alpha}(y;r) = \frac{n^2+(\alpha + 2)n}{y} - \frac{n c_{1}}{2y},
\end{equation}
where $c_{1}$ is given in \eqref{def of c_1}.
\end{lemma}
\begin{proof}
We begin by making the substitution $\lambda_i = y \xi_i$ in \eqref{def of Zn}, which gives 
\begin{equation}\label{Zn tilde}
Z_{n,\alpha}(y;r)  = y^{n^2 + (\alpha + 2) n} \widetilde{Z}_{n,\alpha}(y;r),
\end{equation} 
where
\begin{equation}
\widetilde{Z}_{n,\alpha}(y;r) := \frac{1}{n!} \int_{r}^{\infty}...\int_{r}^{\infty} \Delta_{n}(\xi)^2 \prod_{j=1}^n (\xi_j  - 1)^2 \xi_j^\alpha e^{- n y \xi_j}d\xi_j.
\end{equation}
The above quantity may be computed by introducing monic orthogonal polynomials $\widetilde{p}_j$ satisfying
\begin{equation} \label{ortho p bar}
\int^\infty_r \widetilde{p}_\ell(x) \widetilde{p}_m(x) \widetilde{w}(x) dx = \widetilde{h}_\ell \delta_{\ell m},\qquad \widetilde{w}(x) =(x-1)^2 x^\alpha e^{- n y x}.
\end{equation}
Analogously to \eqref{squared norm}, for $j=0,1,2,...$ we have
\begin{equation} \label{squared norm tilde}
\widetilde{h}_{j} = \frac{\widetilde Z_{j+1,\alpha}(y;r)}{\widetilde Z_{j,\alpha}(y;r)}, \qquad \widetilde Z_{0,\alpha}(y;r) := 1.
\end{equation}
From the orthogonality conditions for $p_j$ given in \eqref{ortho p}, we easily obtain
\begin{equation}
\widetilde{p}_j(x) = y^{-j} p_j(yx), \qquad \widetilde{h}_j = y^{-(2j+\alpha+3)} h_j.
\end{equation}
A similar calculation for the Cauchy transform $q_j$ appearing in \eqref{Ysol} shows that
\begin{equation}
\widetilde{q}_j(z) = \frac{1}{2\pi i} \int^\infty_r \frac{\widetilde{p}_j(x) \widetilde{w}(x)}{x-z} dx = y^{-(j+\alpha+2)} q_j(y z).
\end{equation}
Summarising, if we define
\begin{equation}
\widetilde{Y}(z) := \begin{pmatrix}\widetilde{p}_n(z) & \widetilde{q}_n(z) \\ -\frac{2\pi i}{\widetilde{h}_{n-1}} \widetilde{p}_{n-1}(z)& -\frac{2\pi i}{\widetilde{h}_{n-1}} \widetilde{q}_{n-1}(z)\end{pmatrix},
\end{equation}
we obtain the relationship
\begin{equation}\label{Y tilde in term of Y}
\widetilde{Y}(z) = y^{-n \sigma_3} y^{-\frac{\alpha+2}{2}\sigma_3} Y(yz) y^{\frac{\alpha +2 }{2}\sigma_3}.
\end{equation}
We now use the well known relation, which can be straightforwardly deduced from \eqref{squared norm tilde},
\begin{equation}
\widetilde{Z}_{n,\alpha}(y;r) = \prod_{i=0}^{n-1} \widetilde{h}_i,
\end{equation}
from which it follows that
\begin{equation} \label{d log of Z}
\partial_y \log \widetilde Z_{n,\alpha}(y;r) = \sum_{i=0}^{n-1} \frac{\partial_y \widetilde{h}_i}{\widetilde{h}_i}.
\end{equation}
Note that
\begin{equation}
\partial_y \widetilde{h}_i =\int^\infty_{r} \widetilde{p}_i(x)^2 \partial_y \widetilde{w}(x) dx,
\end{equation}
where in the above line we have used the fact that $\widetilde{p}_{i}$ and $\partial_y \widetilde{p}_{i}$ are orthogonal. Combining the above expressions then yields,
\begin{equation}
\partial_y \widetilde{h}_i = -n \int^\infty_{r} x \widetilde{p}_i(x)^2 \widetilde{w}(x) dx.
\end{equation}
Using the above in \eqref{d log of Z} we have,
\begin{equation}\label{lol18}
\partial_y \log \widetilde{Z}_{n,\alpha}(y;r) =  - n \int^\infty_{r} x \widetilde{w}(x) \sum_{i=0}^{n-1} \frac{\widetilde{p}_i(x)^2}{\widetilde{h}_i} dx.
\end{equation}
The summation can now be removed by use of the Christoffel-Darboux formula (see \cite[equation (3.2.4)]{Szego OP})
\begin{equation}\label{Christofel Darboux}
\sum_{i=0}^{n-1} \frac{\widetilde{p}_i(x)^2}{\widetilde{h}_i} = \frac{\widetilde{p}_{n}(x)^{\prime}\widetilde{p}_{n-1}(x)-\widetilde{p}_{n-1}(x)^{\prime}\widetilde{p}_{n}(x)}{\widetilde{h}_{n-1}}.
\end{equation}
In order to simplify by a contour deformation the integral in the right-hand side of \eqref{lol18}, we will use the formula
\begin{equation}
\widetilde{w}(x) \sum_{i=0}^{n-1} \frac{\widetilde{p}_i(x)^2}{\widetilde{h}_i} = -\frac{1}{4\pi
i}\left(\mathrm{Tr}\left(\widetilde{Y}^{-1}_{+}(x)\widetilde{Y}^{\prime}_{+}(x)\sigma_3 \right) -
\mathrm{Tr}\left(\widetilde{Y}^{-1}_{-}(x)\widetilde{Y}^{\prime}_{-}(x) \sigma_3  \right)\right), \qquad x \in (r,\infty),
\end{equation}
which can be obtained from \eqref{Christofel Darboux} and from the relation $\widetilde{Y}_{+}(x) = \widetilde{Y}_{-}(x) \begin{pmatrix}
1 & \widetilde{w}(x) \\ 0 & 1
\end{pmatrix}$ for $x \in (r,\infty)$, and where $\widetilde{Y}_\pm$ correspond to the limiting values of $\widetilde{Y}$ from above and below $(r,\infty)$. We now obtain,
\begin{equation}\label{lol19}
\partial_y \log \widetilde Z_{n,\alpha}(y;r) =  \frac{n}{4\pi i} \int^\infty_{r}  x \left(\mathrm{Tr}\left(\widetilde{Y}^{-1}_{+}(x)\widetilde{Y}^{\prime}_{+}(x)\sigma_3 \right) -
      \mathrm{Tr}\left(\widetilde{Y}^{-1}_{-}(x)\widetilde{Y}^{\prime}_{-}(x) \sigma_3  \right)\right) dx.
\end{equation}
Note also that \eqref{Y tilde in term of Y} implies that
\begin{equation}\label{lol20}
\mathrm{Tr}(\widetilde{Y}^{-1}(z)\widetilde{Y}^{\prime}(z)\sigma_3) = y\mathrm{Tr}(Y^{-1}(y z)Y^{\prime}(y z)\sigma_3).
\end{equation}
Therefore, combining \eqref{lol19} with \eqref{lol20} gives after a change of variables
\begin{equation}\label{lol21}
\partial_y \log \widetilde Z_{n,\alpha}(y;r) =  \frac{n}{4\pi i y} \int^\infty_{yr}  x \left(\mathrm{Tr}\left({Y}^{-1}_{+}(x){Y}^{\prime}_{+}(x)\sigma_3 \right) -
      \mathrm{Tr}\left({Y}^{-1}_{-}(x){Y}^{\prime}_{-}(x) \sigma_3  \right)\right) dx.
\end{equation}
\begin{figure}[h]
\begin{center}
\setlength{\unitlength}{1truemm}
\begin{picture}(100,55)(-5,10)

\put(50,40){\arc[11.5,348.5]{25}}
\put(50,40){\arc[90,270]{5}}
\put(50,40){\line(1,0){28}}
\put(50,45){\line(1,0){24.5}}
\put(50,35){\line(1,0){24.5}}
\put(50,15){\thicklines\vector(1,0){.0001}}
\put(50,65){\thicklines\vector(-1,0){.0001}}
\put(65,45){\thicklines\vector(1,0){.0001}}
\put(62.5,35){\thicklines\vector(-1,0){.0001}}
\put(45,41.5){\thicklines\vector(0,1){.0001}}
\put(61,47){$\mathcal{C}_{+}$}
\put(61,31.5){$\mathcal{C}_{-}$}

\put(49,36.8){$yr$}\put(50,40){\thicklines\circle*{1.2}}

\put(74,36.8){$yr+R$}\put(75,40){\thicklines\circle*{1.2}}\put(50,67){$\mathcal{C}_{R}$}

\put(44,46.5){$yr+i\epsilon$}\put(50,45){\thicklines\circle*{1.2}}\put(40,39){$\mathcal{C}_{\epsilon}$}
\end{picture}
\caption{The contour $\mathcal{C} = \mathcal{C}_{\epsilon}\cup\mathcal{C}_{+}\cup\mathcal{C}_{R}\cup\mathcal{C}_{-} $ used in establishing a differential identity for the Hankel determinant $Z_{n,\alpha}(y,r)$.\label{diff identity contour}}
\end{center}
\end{figure}
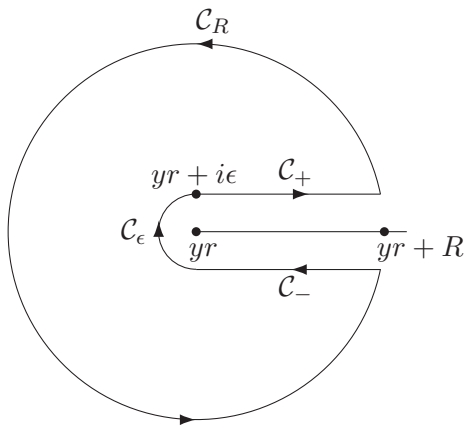
Consider the integral of $z  \mathrm{Tr} \left(Y^{-1}(z)Y^{\prime}(z)  \sigma_3 \right)$ over the contour $\mathcal C$ shown in Figure \ref{diff identity contour}. As $Y$ is analytic in $\mathbb{C}\setminus [yr,\infty)$ and $\mathcal{C}$ does not enclose any singularities of $Y$, this integral is zero. Therefore, we have
\begin{equation}\label{lol22}
\int_{\mathcal{C}_{+}\cup \mathcal{C}_{-}}   z \mathrm{Tr}\left({Y}^{-1}(z){Y}^{\prime}(z)\sigma_3 \right) dz = -\int_{\mathcal{C}_\epsilon+ \mathcal{C}_R}  z \mathrm{Tr}\left({Y}^{-1}(z){Y}^{\prime}(z)\sigma_3 \right)dz.
\end{equation}
Property $(d)$ in the RH problem for $Y$ implies that $Y^{-1}(z)Y^{\prime}(z) = \bigO((\log(yr-z))^{2})$ as $z \to yr$, and thus we have
\begin{equation}
\lim_{\epsilon \to 0} \int_{\mathcal{C}_{\epsilon}} z \mathrm{Tr}\left({Y}^{-1}(z){Y}^{\prime}(z)\sigma_3 \right)dz = 0.
\end{equation}
Thus, by Remark \ref{remark: def of Res} and by taking first $\epsilon \to 0$ and then $R \to \infty$ in \eqref{lol22}, only the term containing $c_{1}$ in \eqref{def of c_1} contributes to the limit, one has
\begin{equation*}
\int^\infty_{yr}  x \left(\mathrm{Tr}\left({Y}^{-1}_{+}(x){Y}^{\prime}_{+}(x)\sigma_3 \right) -
      \mathrm{Tr}\left({Y}^{-1}_{-}(x){Y}^{\prime}_{-}(x) \sigma_3  \right)\right) dx = -2\pi i c_{1},
\end{equation*}
and equation \eqref{lol21} becomes
\begin{equation}
\partial_y \log \widetilde Z_{n,\alpha}(y;r) =  -\frac{nc_{1}}{2 y}.
\end{equation}
The result follows from \eqref{Zn tilde}.
\end{proof}

\section{A Riemann-Hilbert problem related to the system of ODEs}
\label{section RH model Phi}

We first introduce some notations for the sake of convenience. We define the piecewise constant function
\begin{equation}
\theta(z) = \left\{ \begin{array}{l l}
+1 & \mbox{ if }\mbox{ Im}z>0, \\
-1 & \mbox{ if }\mbox{ Im}z<0, \\
\end{array} \right.
\end{equation}
and for $t \in \mathbb{R}$, we define also
\begin{equation}
H_{t}(z) = \left\{  \begin{array}{l l}

I, & \mbox{ for } -\frac{2\pi}{3}< \arg(z-t)< \frac{2\pi}{3},\\

\begin{pmatrix}
1 & 0 \\
-e^{\pi i \alpha} & 1 \\
\end{pmatrix}, & \mbox{ for } \frac{2\pi}{3}< \arg(z-t)< \pi, \\

\begin{pmatrix}
1 & 0 \\
e^{-\pi i \alpha} & 1 \\
\end{pmatrix}, & \mbox{ for } -\pi< \arg(z-t)< -\frac{2\pi}{3}, \\

\end{array} \right.
\end{equation}
where the principal branch is chosen for the argument, such that $\arg(z-t) = 0$ if $z >t$.
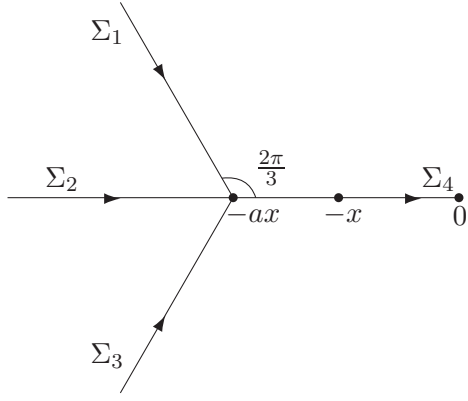
\begin{figure}[h]
    \begin{center}
    \setlength{\unitlength}{1truemm}
    \begin{picture}(100,55)(-5,10)
        \put(50,40){\line(1,0){30}}
        \put(50,40){\line(-1,0){30}}
        \put(50,40){\thicklines\circle*{1.2}}
        \put(50,40){\line(-0.5,0.866){15}}
        \put(50,40){\line(-0.5,-0.866){15}}
        \qbezier(53,40)(52,43)(48.5,42.598)
        \put(53,43){$\frac{2\pi}{3}$}
        \put(49,36.8){$-ax$}
        \put(62,36.8){$-x$}
        \put(79.2,36.4){$0$}
        \put(31,61){$\Sigma_{1}$}
        \put(31,18){$\Sigma_{3}$}
        \put(25,41.5){$\Sigma_{2}$}
        \put(75,41.5){$\Sigma_{4}$}
        \put(75,39.9){\thicklines\vector(1,0){.0001}}
        \put(35,39.9){\thicklines\vector(1,0){.0001}}
        \put(41,55.588){\thicklines\vector(0.5,-0.866){.0001}}
        \put(41,24.412){\thicklines\vector(0.5,0.866){.0001}}
        \put(64,40){\thicklines\circle*{1.2}}
        \put(80,40){\thicklines\circle*{1.2}}
    \end{picture}
    \caption{The jump contours appearing in the Riemann-Hilbert problem for $\Phi$. The arcs $\Sigma_{j}$, $j = 1,2,3,4$ depend on $x$ and $a$. For convenience, they are chosen such that they do not contain $-ax$, $-x$ and $0$. \label{model jump contours}}
\end{center}
\end{figure}
In the course of computing the asymptotics for the Hankel determinants $Z_{n,\alpha}(y;r)$, one is led to consider a model RH problem, which we will denote by $\Phi$ its unique solution. The matrix-valued function $\Phi$ depends on parameters $x > 0$ and $a > 1$. In the RH analysis of $Y$ in Section \ref{Section Steepest descent Y}, $x$ and $a$ will be related to $y$ and $r$ through the relations
\begin{equation}\label{lol17}
x = n^{2}f(y) \qquad \mbox{ and } \qquad a = \frac{f(yr)}{f(y)},
\end{equation}
where $f$ is a conformal map from $0$ to a neighbourhood of $0$, satisfying $f^{\prime}(0) = 4$, see \eqref{relations a r and x y} and \eqref{conformal map at 0}. In particular, if we can write $y = \frac{s}{4n^{2}}$ for a fixed $s$, and if $r$ is fixed, equation \eqref{lol17} implies that as $n \to \infty$, we have
\begin{equation}
x = s+\bigO(n^{-2}), \qquad a = r + \bigO(n^{-2}).
\end{equation}

\subsubsection*{RH problem for $\Phi(z) = \Phi(z;x,a)$}
\label{PhiRH}
\begin{itemize}
\item[(a)] $\Phi:\mathbb C\setminus\Sigma_{x,ax} \to\mathbb C^{2\times 2}$ analytic, with $\Sigma_{x,ax} = \cup^{4}_{i=1} \Sigma_i \cup \{-ax,-x,0\}$ as illustrated in Figure \ref{model jump contours}. 
\item[(b)] $\Phi$ has continuous boundary values $\Phi_\pm(z)$ as $z\in\Sigma_{x,ax}\setminus\{-ax,-x,0 \}$ is approached from the left ($+$) or right ($-$) side of $\Sigma_{x,ax}$, and they are related by
\begin{align}
&\Phi_+(z)=\Phi_-(z)\begin{pmatrix}1&0\\e^{\pi i \alpha}&1\end{pmatrix}, &z\in\Sigma_1,\\
&\Phi_+(z)=\Phi_-(z)\begin{pmatrix}0&1\\-1&0\end{pmatrix}, &z\in \Sigma_2,\\
&\Phi_+(z)=\Phi_-(z)\begin{pmatrix}1&0\\e^{-\pi i \alpha}&1\end{pmatrix}, &z\in\Sigma_3,\\
&\Phi_+(z)=\Phi_-(z) e^{\pi i \alpha \sigma_3}, &z\in\Sigma_4.
\end{align}
\item[(c)] As $z\to\infty$, there exist functions $p(x), q(x)$ and $v(x)$ (these functions also depend on $a$), such that $\Phi$ has the asymptotic behaviour
\begin{equation}\label{Psic}
\Phi(z)=\left(I+\frac{1}{z} \Phi_{1}(x) + \bigO(z^{-2})\right) z^{-\frac{1}{4}\sigma_3}Ne^{z^\frac{1}{2}\sigma_3},
\end{equation}
where $N=\frac{1}{\sqrt{2}}(I+i\sigma_1)$, $\sigma_1=\begin{pmatrix}0&1\\1&0\end{pmatrix}$ and 
\begin{equation*}
\Phi_{1}(x) = \begin{pmatrix}q(x)&i v(x)\\i p(x)&-q(x)\end{pmatrix}.
\end{equation*}
\item[(d)] As $z\to -ax$, $\Phi$ has the asymptotic behaviour
\begin{equation}
\label{Psia}
\Phi(z)=\bigO(1) \begin{pmatrix}
1 & \frac{1}{2\pi i} \log(z+ax) \\ 0 & 1
\end{pmatrix}e^{\frac{\pi i \alpha}{2}\theta(z)\sigma_{3}}H_{-ax}(z).
\end{equation}
As $z\to -x$, $\Phi$ has the asymptotic behaviour
\begin{equation}
\label{Psi1}
\Phi(z)=\bigO(1) e^{\frac{\pi i \alpha}{2}\theta(z)\sigma_{3}} (z+x)^{\sigma_3}.
\end{equation} 
As $z\to 0$, $\Phi$ has the asymptotic behaviour
\begin{equation}
\label{Psi0}
\Phi(z)=\bigO(1) z^{\frac{\alpha\sigma_3}{2}}.
\end{equation}
\end{itemize}

\begin{remark}
It can be verified by deleting the jumps around the singularities that the $\bigO(1)$ terms in asymptotics \eqref{Psia}, \eqref{Psi1} and \eqref{Psi0} are analytic functions.
\end{remark}

\begin{remark}
\label{Remark: uniqueness and existence}
The uniqueness of the solution $\Phi$ follows by standard arguments, based on the fact that $\det \Phi \equiv 1$ (see e.g. \cite[Theorem 7.18]{Deift}). It is in general a more difficult task to prove existence of a given RH problem, this relies on showing a so-called ``vanishing lemma".
The existence of $\Phi$ has been proved for $\alpha \geq 0$ in \cite[Lemma 2.6]{Atkin} (our situation corresponds to $I = (-\infty,-ax)$, $B = \{-x,0\}$, $\hat \alpha_{-x} = 1$, $\hat \alpha_{0} = \frac{\alpha}{2}$ and $\tau_{\infty,0} = -\frac{1}{2}$ in the language of \cite{Atkin}). Nevertheless, the proof of the vanishing lemma \cite[Lemma 2.6]{Atkin} does not require the assumptions $\alpha \geq 0$ and holds more generally for $\alpha >-1$. Thus, $\Phi$ exists and is unique for $\alpha >-1$.
\end{remark}

\begin{remark}
The fact that $\Phi_{1}(x)$ is traceless follows immediately from $\det \Phi \equiv 1$.
\end{remark}

\begin{remark}\label{Remark: symmetry}
Since $\alpha \in \mathbb{R}$, we can check that $\sigma_{3} \overline{\Phi(\overline{z})}\sigma_{3}$ is also a solution of the RH problem for $\Phi$. Thus, by uniqueness of the solution (see Remark \ref{Remark: uniqueness and existence}), we have
\begin{equation}
\Phi(z) = \sigma_{3} \overline{\Phi(\overline{z})}\sigma_{3}.
\end{equation}
In particular, this implies that all the functions $p(x), q(x)$ and $v(x)$ are real.
\end{remark}

\section{The Lax pair and two coupled Painlev\'{e} V equations}\label{Section Lax Pair}
To derive the system of two coupled Painlev\'{e} V equations \eqref{system of q_1 q_2}, we will use a well-known method of isomonodromic deformation theory \cite{FokasItsKapaevNovokshenov}. We begin by making the transformation
\begin{equation}
\widetilde\Phi(z;x,a) := \begin{pmatrix} 1&0\\-\frac{v(x^{2})}{x}&1 \end{pmatrix} x^\frac{\sigma_3}{2}e^{\frac{\pi i}{4} \sigma_3}\Phi(x^2 z;x^2,a).
\end{equation}
The jump contour for $\widetilde\Phi(z)$ is $\Sigma_{1,a}$ and is independent of $x$, while its asymptotic expansion as $z \to \infty$ is given by
\begin{equation}\label{large z asymp for tilde phi}
\widetilde{\Phi}(z)=\begin{pmatrix}
1 & 0 \\ -\frac{v(x^{2})}{x} & 1
\end{pmatrix} \left(I+\frac{1}{z} \widetilde{\Phi}_{1}(x) + \bigO(z^{-2})\right) e^{\frac{\pi i}{4}\sigma_{3}} z^{-\frac{1}{4}\sigma_3}Ne^{x z^\frac{1}{2}\sigma_3},
\end{equation}
where 
\begin{equation}
\widetilde{\Phi}_{1}(x) = 
\begin{pmatrix}
\frac{q(x^{2})}{x^{2}} & -\frac{v(x^{2})}{x} \\  \frac{p(x^{2})}{x^{3}} & -\frac{q(x^{2})}{x^{2}}
\end{pmatrix}.
\end{equation}
By standard arguments of RH analysis, $\tilde{\Phi}$ is analytic in $x \in (0,\infty)$. The Lax pair $(A,B) = (A(z;x,a),B(z;x,a))$ is defined by,
\begin{align}
\partial_z \widetilde\Phi(z) = A(z) \widetilde\Phi(z), \label{A equation in the lax pair}\\
\partial_x \widetilde\Phi(z) = B(z) \widetilde\Phi(z). \label{B equation in the lax pair}
\end{align}
The fact that $\det \widetilde{\Phi}$ is constant implies that $A$ and $B$ are traceless. Also, since the jump matrices for $\widetilde\Phi(z)$ are independent of $z$ and $x$, $A$ and $B$ are analytic on $\mathbb{C}\setminus \{-a,-1,0\}$. Using the asymptotic behaviour of  $\widetilde\Phi(z)$ as $z\to 0$, $z\to -1$, $z\to -a$ and $z\to \infty$, it is easy to show that $A(z)$ is meromorphic on $\mathbb{C}$ with single poles at $-a$, $-1$ and $0$ while $B(z)$ is an entire function. One has
\begin{align}
\label{ALax}
&A(z)=A_{\infty,0}(x) + A_{0,1}(x)z^{-1} + A_{1,1}(x)(z+1)^{-1} + A_{a,1}(x)(z+a)^{-1}, \\[0.3cm]
\label{BLax}
&B(z)=\begin{pmatrix}0&1\\z+u(x)&0\end{pmatrix}, \qquad \mbox{where} \qquad u(x) = \frac{-2 v^{\prime}(x^{2})x^{2}+v(x^{2})^{2}-2q(x^{2})+v(x^{2})}{x^{2}},
\end{align}
the matrices $A_{0,1}(x)$, $A_{1,1}(x)$, $A_{a,1}(x)$ are analytic in $x \in (0,\infty)$ and $A_{\infty,0}(x)=\begin{pmatrix}
0 & 0 \\ \frac{x}{2} & 0
\end{pmatrix}$. There are infinitely many non-trivial relations between the functions appearing in \eqref{large z asymp for tilde phi}. They can be found using the fact that $B$ is entire. For example, by expending $B_{12}(z)$ as $z \to \infty$ using \eqref{large z asymp for tilde phi}, we find
\begin{equation}
B_{12}(z) = 1 + \frac{-v(x^{2})^{2}+v(x^{2})+2q(x^{2})-2x^{2}v^{\prime}(x^{2})}{x^{2}z} + \bigO(z^{-2}), \qquad \mbox{ as } z \to \infty,
\end{equation}
from which we obtain the relation $q(x) = \frac{1}{2} (2 x v^{\prime}(x)  + v^{2}(x) - v(x))$, and thus $u(x)$ can be rewritten more simply as $u(x) = -2\left( v(x^{2})x^{-1} \right)^{\prime}$. We now turn to the compatibility condition
\begin{equation}
\partial_z \partial_x \widetilde\Phi = \partial_x \partial_z \widetilde\Phi,
\end{equation}
which upon rewriting the derivatives in terms of the Lax matrices becomes
\begin{equation} \label{ZCeqn}
\partial_x A - \partial_z B + AB - BA = 0.
\end{equation}
If we parameterise $A$ as
\beq\label{para A}
A(z;x,a) = \begin{pmatrix}d(z;x,a)&b(z;x,a)\\c(z;x,a)&-d(z;x,a)\end{pmatrix},
\eeq
then \eqref{ZCeqn} is equivalent to three coupled ODEs, 
\begin{align}
\label{asol}
&d = -\frac{b^{\prime}}{2} ,\\
\label{csol}
&c = (z+u) b - \frac{b^{\prime\prime}}{2} ,\\
\label{aceqn}
&c^{\prime} = 1 + 2(z+u)d,
\end{align}
where primes denote again derivatives with respect to $x$, and where the dependence of the functions in $z$, $x$ and $a$ have been omitted. The first two equations provide $d$ and $c$ in terms of $b$. Taking the determinant of $A$ yields
\begin{equation}
\label{beqn}
\det A = -\frac{(b^{\prime})^{2}}{4} - (z+u) b^2 + \frac{b b^{\prime\prime}}{2}.
\end{equation}
From \eqref{ALax} we have that $b(z)$ is of the form
\begin{equation}\label{b in the lax pair}
b(z) = b_0 z^{-1} + b_1 (z+1)^{-1} + b_2 (z+a)^{-1},
\end{equation}
where $b_{0}$, $b_{1}$ and $b_{2}$ only depend on $x$ and $a$. Note that Remark \ref{Remark: symmetry} implies that $A(z) \in \mathbb{R}$ and $B(z) \in \mathbb{R}$ for $z \in \mathbb{R}$. In particular, $b_{0}$, $b_{1}$ and $b_{2}$ are real. Equation \eqref{b in the lax pair} together with \eqref{beqn}, allow us to compute the asymptotics of $\det A$ at $z=0$, $z=-1$, $z=-a$ and $z=\infty$ in terms of $b_0$, $b_1$ and $b_2$. Alternatively, we may compute $\det A$ at these four points using the asymptotic expansion of $\widetilde\Phi$. Equating these asymptotics with those expressed in terms of $b_0$, $b_{1}$ and $b_2$ we arrive at the equations
\begin{align}
&\frac{\alpha ^2}{4}-b_0(x){}^2 u(x)-\frac{1}{4} b_0'(x){}^2+\frac{1}{2} b_0(x) b_0''(x)=0, \label{eq diff u}\\
&b_1(x){}^2 (1-u(x))-\frac{1}{4} b_1'(x){}^2+\frac{1}{2} b_1(x) b_1''(x)+1=0,\label{eq diff b1}\\
&b_2(x){}^2 (a-u(x))-\frac{1}{4} b_2'(x){}^2+\frac{1}{2} b_2(x) b_2''(x)=0,\label{eq diff b2}\\
\label{4th of Det A}&\frac{x^2}{4}-\left(b_0(x)+b_1(x)+b_2(x)\right){}^2=0.
\end{align}
By expanding the expression $A_{12}(z) = b(z)$ in a Laurent series about $z=\infty$, we get the following identities between $b_{0}$, $b_{1}$, $b_{2}$ and $v^{\prime}$:
\begin{align}
& b_{0}(x) + b_{1}(x) + b_{2}(x) = \frac{x}{2}, \label{system for the b first} \\[0.3cm]
& b_{1}(x) + a b_{2}(x) = - x v^{\prime}(x^{2}). \label{system for the b second}
\end{align}
Note that \eqref{system for the b first} is a better version of \eqref{4th of Det A}. If we express $b_{0}$ and $u$ in terms of $b_{1}$ and $b_{2}$ from \eqref{eq diff u} and  \eqref{system for the b first}, and if we define 
\begin{equation}\label{change of functions}
q_{j}^{2}(x) = \frac{2b_{j}(\sqrt{x})}{\sqrt{x}}, \qquad j = 1,2,
\end{equation}
we obtain the system \eqref{system of q_1 q_2} with $a=r$ from \eqref{eq diff b1} and \eqref{eq diff b2}. The same change of functions \eqref{change of functions} was used in \cite{CharlierDoeraene}, where the authors obtained a system of $k$ ($k \in \mathbb{N}_{0}$) coupled Painlev\'{e} V equations.

\section{Further properties of the special solutions}\label{Section asymptotics of the special solutions}

Our main goal in this section is to get asymptotics for $b_{0}(x)$, $b_{1}(x)$ and $b_{2}(x)$ as $(a-1)x \to \infty$ and $ax \to 0_{+}$. 

\subsection{Asymptotic analysis when $(a-1)x \rightarrow \infty$}

\subsubsection{Re-scaling of the model problem}
\label{Subsection: Re-scaling of the model}

In order to have a jump contour independent of $x$, we make the transformation $C(z;x,a) = (ax)^{\frac{1}{4} \sigma_{3}} \Phi(axz;x,a)$. $C$ is the solution of the following RH problem:

\subsubsection*{RH problem for $U$}
\begin{itemize}
\item[(a)] $C : \mathbb{C} \setminus \Sigma_{a^{-1},1} \rightarrow \mathbb{C}^{2 \times 2}$ is analytic.
\item[(b)] $C$ has the following jumps
\begin{align}
&C_+(z)=C_-(z)\begin{pmatrix}1&0\\e^{\pi i \alpha}&1\end{pmatrix}, &z\in\Sigma_1,\\
&C_+(z)=C_-(z)\begin{pmatrix}0&1\\-1&0\end{pmatrix}, &z\in \Sigma_2,\\
&C_+(z)=C_-(z)\begin{pmatrix}1&0\\e^{-\pi i \alpha}&1\end{pmatrix}, &z\in\Sigma_3,\\
&C_+(z)=C_-(z) e^{\pi i \alpha \sigma_3}, &z\in\Sigma_4.
\end{align}
\item[(c)] As $z \rightarrow \infty$,
\begin{equation}
C(z) = \left(I+\frac{1}{z}C_{1}(x;a)+ \bigO(z^{-2})\right) z^{-\frac{1}{4}\sigma_3}Ne^{\sqrt{ax}z^{\frac{1}{2}}\sigma_3},
\end{equation}
where 
\begin{equation}
C_{1}(x;a) = \begin{pmatrix}
\frac{q(x;a)}{ax} & \frac{iv(x;a)}{(ax)^{1/2}} \\
\frac{ip(x;a)}{(ax)^{3/2}} & -\frac{q(x;a)}{ax}
\end{pmatrix}.
\end{equation}
\item[(d)] As $z\to -1$, $C$ has the asymptotic behaviour
\begin{equation}
C(z) = \bigO(1) \begin{pmatrix}
1 & \frac{1}{2\pi i} \log (z+1) \\ 0 & 1
\end{pmatrix}e^{\frac{\pi i \alpha}{2}\theta(z)\sigma_{3}}H_{-1}(z).
\end{equation}
As $z\to -a^{-1}$, $C$ has the asymptotic behaviour
\begin{equation}
C(z)=\bigO(1)e^{\frac{\pi i \alpha}{2}\theta(z)\sigma_{3}}\left(z+a^{-1}\right)^{\sigma_3}.
\end{equation}
As $z\to 0$, $C$ has the asymptotic behaviour
\begin{equation}
C(z)=\bigO(1)z^{\frac{\alpha\sigma_3}{2}}.
\end{equation}
\end{itemize}

\subsubsection{Normalisation at $\infty$ of the RH problem}

We define the $g$-function by 
\begin{equation}
g(z) = \sqrt{z+1}
\end{equation}
where the principal branch is taken for the square root. As $z \to \infty$, $g$ has the asymptotic behaviour
\begin{equation}
g(z) = z^{1/2} + g_{1} z^{-1/2} + \bigO(z^{-3/2}), \qquad g_{1} = \frac{1}{2}.
\end{equation}
Now we define 
\begin{equation}\label{lol12}
W(z) = \begin{pmatrix}
1 & 0 \\ ig_{1} \sqrt{ax} & 1
\end{pmatrix} C(z) e^{-\sqrt{ax}g(z)\sigma_{3}}.
\end{equation}

\subsubsection*{RH problem for $W$}
\begin{itemize}
\item[(a)] $W : \mathbb{C} \setminus \Sigma_{a^{-1},1} \rightarrow \mathbb{C}^{2 \times 2}$ is analytic.
\item[(b)] $W$ has the following jumps:
\begin{align}
&W_{+}(z)=W_{-}(z) \begin{pmatrix} 1&0\\e^{\pi i \alpha}e^{-2\sqrt{ax}g(z)}&1\end{pmatrix}, &z\in\Sigma_1, \label{Jump1 of W}\\
&W_{+}(z) = W_{-}(z) \begin{pmatrix}0&1\\-1&0\end{pmatrix}, &z\in \Sigma_2,\\
&W_{+}(z)=W_{-}(z)\begin{pmatrix}1&0\\e^{-\pi i \alpha}e^{-2\sqrt{ax}g(z)}&1\end{pmatrix}, &z\in\Sigma_3, \label{Jump3 of W}\\
&W_{+}(z)=W_{-}(z) e^{\pi i \alpha \sigma_3}, &z\in\Sigma_4.
\end{align}
\item[(c)] As $z \rightarrow \infty$,
\begin{equation} \label{W asymp inf}
W(z) = \left(I+\frac{1}{z}W_{1}(x;a) + \bigO(z^{-2})\right) z^{-\frac{1}{4}\sigma_3} N,
\end{equation}
where 
\begin{equation} \label{W1 equation}
(W_{1}(x;a))_{12} = \frac{i v(x;a)}{\sqrt{ax}} + i g_{1}\sqrt{ax}.
\end{equation}
\item[(d)] As $z\to -1$, 
\begin{equation}
W(z)=\bigO(1) \begin{pmatrix}
1 & \frac{\log(z+1)}{2\pi i} \\ 0 & 1
\end{pmatrix}
e^{ \frac{\pi i \alpha}{2}\theta(z)\sigma_{3}} H_{-1}(z)e^{-\sqrt{ax}g(z)\sigma_{3}}.
\end{equation}
As $z\to -a^{-1}$,
\begin{equation}
W(z)=\bigO(1)e^{\frac{\pi i \alpha}{2}\theta(z)\sigma_{3}}\left(z+a^{-1} \right)^{\sigma_3}.
\end{equation}
As $z\to 0$,
\begin{equation}
W(z)=\bigO(1)z^{\frac{\alpha\sigma_3}{2}}.
\end{equation}
\end{itemize}
For $z \in \Sigma_{1} \cup \Sigma_{3}$, $\mbox{Re(}g(z)\mbox{)} > 0$ and therefore the jumps of $W$ on $\Sigma_{1} \cup \Sigma_{3}$ are exponentially close to the identity matrix as $(a-1)x \to \infty$. Since $g(-1)=0$, this convergence is not uniform as $z$ approaches $-1$. Therefore we will construct a global parametrix which will be a good approximation of $W$ as $z$ stays away from a neighbourhood of $-1$, and a local parametrix around $-1$.

\subsubsection{Global parametrix}

Ignoring exponentially small entries in the jumps and a small neighbourhood of $-1$, we are led to consider the following RH problem.

\subsubsection*{RH problem for $P^{(\infty)}$}
\begin{itemize}
\item[(a)] $P^{(\infty)} : \mathbb{C} \setminus \mathbb{R}^{-}$ is analytic.
\item[(b)] $P^{(\infty)}$ has the following jumps on $\mathbb{R}^{-}$:
\begin{align}
&P^{(\infty)}_{+}(z) = P^{(\infty)}_{-}(z) \begin{pmatrix}0&1\\-1&0\end{pmatrix}, &z\in (-\infty,-1),\\
&P^{(\infty)}_{+}(z) = P^{(\infty)}_{-}(z) e^{\pi i \alpha \sigma_3}, & z \in (-1,0)\setminus \{-a^{-1}\},
\end{align}
\item[(c)] As $z \rightarrow \infty$,
\begin{equation} \label{Pinf asymp inf}
P^{(\infty)}(z) = \left(I+\frac{1}{z}P^{(\infty)}_{1}(a) + \bigO(z^{-2})\right) z^{-\frac{1}{4}\sigma_3} N.
\end{equation}
\item[(d)] As $z\to -1$,
\begin{equation}
P^{(\infty)}(z)=\bigO\left((z+1)^{-1/4}\right).
\end{equation}
As $z\to -a^{-1}$,
\begin{equation}
P^{(\infty)}(z)=\bigO(1)e^{\frac{\pi i \alpha}{2}\theta(z)\sigma_{3}}\left(z+a^{-1}\right)^{\sigma_3}.
\end{equation}
As $z\to 0$,
\begin{equation}
P^{(\infty)}(z)=\bigO(1) z^{\frac{\alpha\sigma_3}{2}}.
\end{equation}
\end{itemize}
We can check that the solution of this RH problem is explicitly given by
\begin{multline}\label{lol11}
P^{(\infty)}(z) = \begin{pmatrix}
1 & 0 \\ i(\alpha + 2 \sqrt{1-a^{-1}}) & 1
\end{pmatrix} (z+1)^{-\frac{\sigma_{3}}{4}} N \\ \times \left( \frac{\sqrt{z+1}+1}{\sqrt{z+1}-1} \right)^{-\frac{\alpha}{2}\sigma_{3}} \left( \frac{\sqrt{z+1} + \sqrt{1-a^{-1}}}{\sqrt{z+1} - \sqrt{1-a^{-1}}} \right)^{-\sigma_{3}},
\end{multline}
where the principal branch has been chosen for each root. Note that 
\begin{equation}\label{lol24}
(P^{(\infty)}_{1}(a))_{12} = i(\alpha + 2 \sqrt{1-a^{-1}}).
\end{equation}
\subsubsection{Local parametrix near $-1$}

We want to construct a function $P^{(-1)}$ defined in a open disk $D_{-1}$ around $-1$ of radius $\frac{1}{3}(1-a^{-1})$ which satisfies the following RH conditions.

\subsubsection*{RH problem for $P^{(-1)}$}
\begin{itemize}
\item[(a)] $P^{(-1)} : D_{-1} \setminus \Sigma_{a^{-1},1}$ is analytic.
\item[(b)] $P^{(-1)}$ has the following jumps:
\begin{align}
&P^{(-1)}_{+}(z)=P^{(-1)}_{-}(z) \begin{pmatrix} 1&0\\e^{\pi i \alpha}e^{-2\sqrt{ax}g(z)}&1\end{pmatrix}, &z\in\Sigma_1 \cap D_{-1}, \label{Jump1 of P^{(-a)}}\\
&P^{(-1)}_{+}(z) = P^{(-1)}_{-}(z) \begin{pmatrix}0&1\\-1&0\end{pmatrix}, &z\in \Sigma_2\cap D_{-1},\\
&P^{(-1)}_{+}(z)=P^{(-1)}_{-}(z)\begin{pmatrix}1&0\\e^{-\pi i \alpha}e^{-2\sqrt{ax}g(z)}&1\end{pmatrix}, &z\in\Sigma_3\cap D_{-1}, \label{Jump3 of P^{(-a)}}\\
&P^{(-1)}_{+}(z)=P^{(-1)}_{-}(z) e^{\pi i \alpha \sigma_3}, &z\in\Sigma_4\cap D_{-1}.
\end{align}
\item[(c)] As $x \rightarrow \infty$,
\begin{equation} \label{P(-a) on the boundary}
P^{(-1)}(z) = \left(I + \bigO\left(((a-1)x)^{-1/2}\right)\right) P^{(\infty)}(z)
\end{equation}
uniformly for $z\in \partial D_{-1}$.
\item[(d)] As $z\to -1$,
\begin{equation}
P^{(-1)}(z) = \bigO(1) \begin{pmatrix}
1 & \frac{\log(z+1)}{2\pi i} \\ 0 & 1
\end{pmatrix}
e^{\frac{\pi i \alpha}{2}\theta(z)\sigma_{3}} H_{-1}(z)e^{-\sqrt{ax}g(z)\sigma_{3}}.
\end{equation}
\end{itemize}
The solution of this RH problem can be constructed in terms of the Bessel model RH problem with parameter $\alpha = 0$, which is presented in the appendix (see Subsection \ref{Subsection: Bessel model RH problem}), and whose solution is denoted $\Upsilon^{(0)}$. The local parametrix is given by
\begin{equation}
P^{(-1)}(z) = E(z) \Upsilon^{(0)}(axf(z)) e^{\frac{\pi i \alpha}{2}\theta(z) \sigma_{3}} e^{-\sqrt{ax}g(z)\sigma_{3}},
\end{equation}
where
\begin{equation}
f(z) = g(z)^{2} = z+1 \quad \mbox{ and } \quad E(z) = P^{(\infty)}(z)e^{-\frac{\pi i \alpha}{2}\theta(z)\sigma_{3}}N^{-1} f(z)^{\frac{\sigma_{3}}{4}}(ax)^{\frac{1}{4}\sigma_{3}}.
\end{equation}
It can be verified that $E$ is analytic in $D_{-1}$.

\subsubsection{Small norm RH problem} \label{Small norm RH problem for x small}

Define
\begin{equation}\label{lol10}
R(z) = \left\{ \begin{array}{l l}
W(z)P^{(\infty)}(z)^{-1}, & \mbox{for } z \in \mathbb{C}\setminus (\overline{D_{-1}} \cup \Sigma_{1} \cup \Sigma_{3}), \\[0.1cm]
W(z)P^{(-1)}(z)^{-1}, & \mbox{for } z \in D_{-1}.
\end{array} \right.
\end{equation}
Since $W$ and $P^{(-1)}$ have the same jumps inside $D_{-1}$ and the same behaviour near $-1$, $R$ is analytic inside $D_{-1}$. Also, $W$ and $P^{(\infty)}$ have the same jumps on $\mathbb{R}^{-}$, and the same behaviour near $-a^{-1}$ and $0$. Therefore, $R$ is analytic on $\mathbb{C}\setminus ((\partial D_{-1} \cup \Sigma_{1} \cup \Sigma_{3}) \setminus D_{-1})$. Let us put the clockwise orientation on $\partial D_{-1}$. On $\partial D_{-1}$ by \eqref{P(-a) on the boundary}, we have $R_{-}(z)^{-1}R_{+}(z) = I + \bigO\left(((a-1)x)^{-1/2}\right)$ and by \eqref{Jump1 of W} and \eqref{Jump3 of W}, on $(\Sigma_{1} \cup \Sigma_{3})\setminus D_{-1}$, $R_{-}(z)^{-1}R_{+}(z) = I + \bigO(e^{-c \sqrt{(a-1)x}})$ where $c > 0$ is a constant. From \eqref{W asymp inf} and \eqref{Pinf asymp inf}, as $z \to \infty$ one has $R(z) = I + \bigO(z^{-1})$. By small norm theory for RH problems, it follows that $R$ exists for sufficiently large $(a-1)x$ and satisfies $R(z) = I + \bigO(((a-1)x)^{-1/2})$ uniformly in $z \in \mathbb{C}\setminus ((\partial D_{-1} \cup \Sigma_{1} \cup \Sigma_{3}) \setminus D_{-1})$ and as $z \to \infty$, 
\begin{equation}\label{lol37}
R(z) = I + \frac{R_{1}(x;a)}{z} + \bigO(z^{-2}), \qquad R^{\prime}(z) = \bigO\left(((a-1)x)^{-1/2}\right)
\end{equation}
We have in particular that $R_{1}(x;a) = \bigO(((a-1)x)^{-1/2})$. For $z \in \mathbb{C}\setminus (\overline{D_{-1}} \cup \Sigma_{1} \cup \Sigma_{3})$, we have $W(z) = R(z) P^{(\infty)}(z)$ and therefore 
\begin{equation} \label{W1 R P}
W_{1}(x;a) = R_{1}(x;a) + P_{1}^{(\infty)}(a) = P_{1}^{(\infty)}(a) + \bigO(((a-1)x)^{-1/2}),
\end{equation}
Using \eqref{W1 equation}, \eqref{lol24} and \eqref{W1 R P}, we obtain 
\begin{equation}\label{v asymptotics for x to inf}
v(x) = -\frac{1}{2}ax + (\alpha + 2 \sqrt{1-a^{-1}}) \sqrt{ax} + \bigO((1-a^{-1})^{-1/2}), \qquad \mbox{ as } (a-1)x \to \infty.
\end{equation}
For the asymptotics for $b_{0}$ and $b_{1}$, from \eqref{A equation in the lax pair}, \eqref{para A} and \eqref{b in the lax pair}, we can use
\begin{align}
& b_{0}(\sqrt{x}) = \lim_{z \to 0} i \sqrt{x} z \left[ \partial_{z}\Phi(axz;x,a)\Phi^{-1}(axz;x,a) \right]_{12}, \label{b0 large} \\
& b_{1}(\sqrt{x}) = \lim_{z \to -a^{-1}} i \sqrt{x} (z+a^{-1}) \left[ \partial_{z}\Phi(axz;x,a)\Phi^{-1}(axz;x,a) \right]_{12}. \label{b1 large}
\end{align}
To compute these limits, we will need the global parametrix. For $z \in \mathbb{C}\setminus (\overline{D_{-1}} \cup \Sigma_{1} \cup \Sigma_{3})$, by the definition of $C$ given at the beginning of Subsection \ref{Subsection: Re-scaling of the model}, \eqref{lol12} and \eqref{lol10}, we have
\begin{equation}
\Phi(axz;x,a) = (ax)^{-\frac{\sigma_{3}}{4}}\begin{pmatrix}
1 & 0 \\ - \frac{i}{2}\sqrt{ax} & 1
\end{pmatrix}R(z)P^{(\infty)}(z)e^{\sqrt{ax}g(z)\sigma_{3}}.
\end{equation}
Using the definition of the global parametrix given by \eqref{lol11}, together with the asymptotics for $R$ \eqref{lol37} and the above equation, the limits \eqref{b0 large} and \eqref{b1 large} are straightforward to compute. As $(a-1)x \to \infty$, we find
\begin{align}
& b_{0}(\sqrt{x}) = \frac{\alpha}{2\sqrt{a}} \left( 1+ \bigO\left(((a-1)x)^{-1/2}\right) \right), \\
& b_{1}(\sqrt{x}) = \frac{1}{\sqrt{a-1}}\left( 1+ \bigO\left(((a-1)x)^{-1/2}\right) \right). \label{lol27}
\end{align}
Asymptotics for $b_{2}$ can be obtained directly from the relation \eqref{system for the b first}, as $(a-1)x \to \infty$ we have
\begin{equation}
b_{2}(\sqrt{x}) = \frac{\sqrt{x}}{2}-\frac{\alpha}{2\sqrt{a}}-\frac{1}{\sqrt{a-1}}+\bigO \bigg( \frac{1}{(a-1)\sqrt{x}} \bigg). \label{lol29}
\end{equation}
Large $(a-1)x$ asymptotics for $q_{1}(x)^{2}$ and $q_{2}(x)^{2}$ are immediate to obtain from \eqref{change of functions}, \eqref{lol27} and  \eqref{lol29} and are given in \eqref{asymp q1 inf} and \eqref{asymp q2 inf} with $a = r$.

\vspace{0.2cm}\hspace{-0.6cm}We will need later the asymptotics for $\Phi(z;x,a)$ when $z \to \infty$ and simultaneously $(a-1)x \to \infty$. Note that the global parametrix $P^{(\infty)}$ defined in \eqref{lol11} only depends on $a$, and is such that its behaviour at $\infty$ \eqref{Pinf asymp inf} has the form
\begin{equation}
P^{(\infty)}(z) = (I+\bigO(z^{-1}))z^{-\frac{1}{4}\sigma_{3}}N, \qquad \mbox{ as } z \to \infty \quad \mbox{ and } \quad (a-1)x \to \infty.
\end{equation}
Thus from \eqref{lol12}, \eqref{lol10} and the definition of $C$ given at the beginning of Subsection \ref{Subsection: Re-scaling of the model}, we have as $\frac{z}{ax} \to \infty$ and simultaneously $(a-1)x \to \infty$ that
\begin{equation}\label{lol13}
\begin{array}{r c l}
\Phi(z;x,a) & = & (ax)^{-\frac{\sigma_{3}}{4}} \begin{pmatrix}
1 & 0 \\ -\frac{i}{2}\sqrt{ax} & 1
\end{pmatrix} \left( I + \bigO\left( \frac{ax}{z} \right) \right)\left( \frac{z}{ax} \right)^{-\frac{\sigma_{3}}{4}}Ne^{\sqrt{z+ax}\sigma_{3}}, \\
 & = & z^{-\frac{\sigma_{3}}{4}}\begin{pmatrix}
 1 & 0 \\ -\frac{i}{2}\frac{ax}{\sqrt{z}} & 1
\end{pmatrix} \left( I+\begin{pmatrix}
\bigO(\tfrac{ax}{z}) & \bigO(\sqrt{\tfrac{ax}{z}}) \\
\bigO\big(\hspace{-0.10cm}\left(\tfrac{ax}{z}\right)^{\frac{3}{2}}\hspace{-0.10cm}\big) & \bigO(\tfrac{ax}{z})
\end{pmatrix}  \right) N e^{\sqrt{z+ax}\sigma_{3}}.
\end{array}
\end{equation}
Furthermore, if we assume that $\frac{\sqrt{z}}{ax}\to \infty$, we have
\begin{equation}
Ne^{\sqrt{z+ax}\sigma_{3}} = \left( I + \begin{pmatrix}
\bigO \Big( \big(\tfrac{ax}{\sqrt{z}}\big)^{2} \Big) & \bigO \left( \tfrac{ax}{\sqrt{z}} \right) \\
\frac{i}{2}\frac{ax}{\sqrt{z}}+\bigO \Big( \big(\tfrac{ax}{\sqrt{z}}\big)^{3} \Big) & \bigO \Big( \big(\tfrac{ax}{\sqrt{z}}\big)^{2} \Big)
\end{pmatrix} \right) Ne^{\sqrt{z}\sigma_{3}}.
\end{equation}
Thus, if $z \to \infty$, $(a-1)x \to \infty$ and simultaneously $\frac{\sqrt{z}}{ax} \to \infty$, \eqref{lol13} becomes
\begin{equation}\label{lol16}
\Phi(z;x,a) = z^{-\frac{\sigma_{3}}{4}} \left( I + \begin{pmatrix}
\bigO \Big( \big(\tfrac{ax}{\sqrt{z}}\big)^{2} \Big) & \bigO \left( \tfrac{ax}{\sqrt{z}} \right) \\
\bigO \Big( \big(\tfrac{ax}{\sqrt{z}}\big)^{3} \Big) & \bigO \Big( \big(\tfrac{ax}{\sqrt{z}}\big)^{2} \Big)
\end{pmatrix} \right) Ne^{\sqrt{z}\sigma_{3}}.
\end{equation}

\subsection{Asymptotic analysis when $ax \rightarrow 0$}

In order to have the rays $\Sigma_{1}$ and $\Sigma_{3}$ of the jump contour independent of $a$ and $x$, we make the following transformation on $\Phi$:
\begin{equation}\label{W tilde def}
\widetilde{W}(z) = \Phi(z) H_{-ax}(z)^{-1}H_{0}(z).
\end{equation}
It is easy to verify that $\widetilde{W}$ satisfies the following RH problem.
\subsubsection*{RH problem for $\widetilde{W}$}
\begin{itemize}
\item[(a)] $\widetilde{W} : \mathbb{C} \setminus \Sigma_{0,0} \rightarrow \mathbb{C}^{2 \times 2}$ is analytic.
\item[(b)] $\widetilde{W}$ has the following jumps
\begin{align}
&\widetilde{W}_+(z) = \widetilde{W}_-(z)\begin{pmatrix}1&0\\e^{\pi i \alpha}&1\end{pmatrix}, &\arg(z) = \frac{2\pi}{3},\\
&\widetilde{W}_+(z) = \widetilde{W}_-(z)\begin{pmatrix}0&1\\-1&0\end{pmatrix}, &z\in (-\infty,-ax),\\
&\widetilde{W}_+(z) = \widetilde{W}_-(z)\begin{pmatrix}1&0\\e^{-\pi i \alpha}&1\end{pmatrix}, &\arg(z) = -\frac{2\pi}{3},\\
&\widetilde{W}_+(z) = \widetilde{W}_-(z) \begin{pmatrix}
e^{\pi i \alpha} & 0 \\ -2 & e^{-\pi i \alpha}
\end{pmatrix}, &z \in (-ax,0)\setminus \{-x\}.
\end{align} 
\item[(c)] As $z \rightarrow \infty$,
\begin{equation} \label{Asympt tilde Phi z to inf}
\widetilde{W}(z) = \left(I+\frac{1}{z}\Phi_{1}(x) + \bigO(z^{-2})\right) z^{-\frac{1}{4}\sigma_3}Ne^{z^{\frac{1}{2}}\sigma_3}.
\end{equation}
\item[(d)] As $z\to -ax$, $\widetilde{W}$ has the asymptotic behaviour
\begin{equation}\label{W small x at -ax}
\widetilde{W}(z)=\bigO(1) \begin{pmatrix}
1 & \frac{\log(z+ax)}{2\pi i} \\ 0 & 1
\end{pmatrix} e^{\frac{\pi i \alpha}{2}\theta(z)\sigma_{3}}H_{0}(z).
\end{equation}
As $z\to -x$,
\begin{equation}\label{W small x at -x}
\widetilde{W}(z)=\bigO(1)e^{\frac{\pi i \alpha}{2}\theta(z)\sigma_{3}}(z+x)^{\sigma_3}H_{0}(z).
\end{equation}
As $z\to 0$, 
\begin{equation}\label{W small x at 0}
\widetilde{W}(z)=\bigO(1)z^{\frac{\alpha\sigma_3}{2}}H_{0}(z).
\end{equation}
In equations \eqref{W small x at -ax}, \eqref{W small x at -x} and \eqref{W small x at 0}, the $\bigO(1)$ are analytic functions in a neighbourhood of their respective point.
\end{itemize}

\subsubsection{Global parametrix}
As $ax \to 0$, the length of $(-ax,0)$ tends to $0$ and the pole at $-x$, the algebraic singularity at $0$, as well as the logarithmic singularity at $-ax$, are merging together. Therefore, for $z$ outside of a neighbourhood of $0$, we expect that the Bessel model RH problem of order $\alpha + 2$ (presented in the appendix, see Subsection \ref{Subsection: Bessel model RH problem}) will be relevant to construct the global parametrix $P^{(\infty)}$. In a small neighbourhood of $0$ and we will construct a new local parametrix around the origin.

\subsubsection*{RH problem for $P^{(\infty)}$}
\begin{itemize}
\item[(a)] $P^{(\infty)} : \mathbb{C} \setminus \Sigma_{0,0}$ is analytic.
\item[(b)]  $P^{(\infty)}$ has the following jumps on $\Sigma_{0,0}$:
\begin{align}
&P^{(\infty)}_{+}(z) = P^{(\infty)}_{-}(z) \begin{pmatrix}0&1\\-1&0\end{pmatrix}, &z\in \mathbb{R}^{-},\\
&P^{(\infty)}_{+}(z) = P^{(\infty)}_{-}(z) \begin{pmatrix}
1 & 0 \\ e^{\pi i \alpha} & 1
\end{pmatrix}, & \arg(z) = \frac{2\pi}{3},\\
&P^{(\infty)}_{+}(z) = P^{(\infty)}_{-}(z) \begin{pmatrix}
1 & 0 \\ e^{-\pi i \alpha} & 1
\end{pmatrix}, & \arg(z) = -\frac{2\pi}{3}.
\end{align}
\item[(c)] As $z \rightarrow \infty$,
\begin{equation} \label{Asymp x small Pinf z to inf}
P^{(\infty)}(z) = \left(I+ \frac{P_{1}^{(\infty)}}{z}+\bigO(z^{-2})\right) z^{-\frac{1}{4}\sigma_3} N e^{z^{\frac{1}{2}} \sigma_{3}}.
\end{equation}
\item[(d)] As $z \to 0$, 
\begin{equation}
P^{(\infty)}(z)=  \left\{ \begin{array}{l l}

\bigO \begin{pmatrix}
|z|^{\frac{\alpha + 2}{2}} & |z|^{-\frac{\alpha + 2}{2}} \\
|z|^{\frac{\alpha + 2}{2}} & |z|^{-\frac{\alpha + 2}{2}}
\end{pmatrix}, & \mbox{ for } -\frac{2\pi}{3} < \arg (z) < \frac{2\pi}{3}, \\

\bigO \begin{pmatrix}
|z|^{-\frac{\alpha + 2}{2}} & |z|^{-\frac{\alpha + 2}{2}} \\
|z|^{-\frac{\alpha + 2}{2}} & |z|^{-\frac{\alpha + 2}{2}}
\end{pmatrix}, & \mbox{ for } \arg(z) \in (-\pi,-\frac{2\pi}{3}) \cup (\frac{2\pi}{3},\pi). \\
\end{array} \right.
\end{equation}
\end{itemize}
The only solution of this RH problem is well-known and given by
\begin{equation}\label{lol9}
P^{(\infty)}(z) = \Upsilon^{(\alpha+2)}(z),
\end{equation}
where $\Upsilon^{(\alpha+2)}$ is the solution of the Bessel model RH problem, presented in Subsection \ref{Subsection: Bessel model RH problem}.
Note that if we don't specify condition (d) in the RH problem for $P^{(\infty)}$, the solution is not unique. From a mathematical point of view, we remark that $\Upsilon^{(\alpha)}$ or $\Upsilon^{(\alpha+4)}$ for example could have also been a suitable choice for $P^{(\infty)}$, but $\Upsilon^{(\alpha+2)}$ is the only one which satisfies condition (d) and which allows us to create a local parametrix around $0$ respecting the matching condition \eqref{Matching Cond P0}. By \eqref{asymptotics_behaviour_of_Bessel_order1}, we have
\begin{align}
& (P_{1}^{(\infty)})_{12} = \frac{i}{8} (4(\alpha+2)^{2}-1). \label{lol30}
\end{align}
In the construction of the local parametrix in a neighbourhood of $0$, we will need a more explicit knowledge of the behaviour of $P^{(\infty)}$ at the origin. It can be verified (see \cite{AtkClaMez}) that $P^{(\infty)}$ can be written as
\begin{equation}\label{lol3}
P^{(\infty)}(z) = P^{(\infty)}_{0}(z) z^{\frac{\alpha+2}{2}\sigma_{3}} \begin{pmatrix}
1 & h(z) \\ 0 & 1
\end{pmatrix} H_{0}(z), \qquad z \in \mathbb{C}\setminus \Sigma_{0,0},
\end{equation}
where $P^{(\infty)}_{0}(z) = P^{(\infty)}_{0,\alpha}(z)$ is an entire function in $z$ for every $\alpha$ while 
\begin{equation}
h(z) = \left\{ \begin{array}{ll}
\frac{1}{2i \sin(\pi \alpha)}, & \alpha \notin \mathbb{Z}, \\
\frac{(-1)^{\alpha}}{2\pi i} \log z, & \alpha \in \mathbb{Z}. \\
\end{array}  \right.
\end{equation}

\subsubsection{Local parametrix near $0$}

We want to construct a function $P^{(0)}$ defined in a fixed open disk $D_{0}$ around $0$ which satisfies exactly the same RH conditions than $\widetilde{W}$ on $D_{0}$ and matches with $P^{(\infty)}$ on $\partial D_{0}$.

\subsubsection*{RH problem for $P^{(0)}$}
\begin{itemize}
\item[(a)] $P^{(0)} : D_{0} \setminus \Sigma_{0,0}$ is analytic.
\item[(b)] $P^{(0)}$ has the following jumps
\begin{align}
&P^{(0)}_+(z) = P^{(0)}_-(z)\begin{pmatrix}1&0\\e^{\pi i \alpha}&1\end{pmatrix}, &z \in \left\{ \arg(z) = \frac{2\pi}{3} \right\} \cap D_{0},\\
&P^{(0)}_+(z) = P^{(0)}_-(z)\begin{pmatrix}0&1\\-1&0\end{pmatrix}, &z\in (-\infty,-ax)\cap D_{0},\\
&P^{(0)}_+(z) = P^{(0)}_-(z)\begin{pmatrix}1&0\\e^{-\pi i \alpha}&1\end{pmatrix}, &z\in \left\{\arg(z) = -\frac{2\pi}{3}\right\}\cap D_{0},\\
&P^{(0)}_+(z) = P^{(0)}_-(z) \begin{pmatrix}
e^{\pi i \alpha} & 0 \\ -2 & e^{-\pi i \alpha}
\end{pmatrix}, &z \in (-ax,0)\setminus \{-x\}.
\end{align} 
\item[(c)] As $ax \to 0$,
\begin{equation} \label{Matching Cond P0}
P^{(0)}(z) = (I + \bigO(ax)) P^{(\infty)}(z)
\end{equation}
uniformly for $z\in \partial D_{0}$.
\item[(d)] As $z\to -ax$, $P^{(0)}$ has the asymptotic behaviour
\begin{equation}\label{small x at -ax}
P^{(0)}(z)=\bigO(1) \begin{pmatrix}
1 & \frac{\log(z+ax)}{2\pi i} \\ 0 & 1
\end{pmatrix} e^{\frac{\pi i \alpha}{2}\theta(z)\sigma_{3}}H_{0}(z).
\end{equation}
As $z\to -x$,
\begin{equation}\label{small x at -x}
P^{(0)}(z)=\bigO(1)e^{\frac{\pi i \alpha}{2}\theta(z)\sigma_{3}}(z+x)^{\sigma_3}H_{0}(z).
\end{equation}
As $z\to 0$, 
\begin{equation}\label{small x at 0}
P^{(0)}(z)=\bigO(1)z^{\frac{\alpha\sigma_3}{2}}H_{0}(z).
\end{equation}
\end{itemize}
We can check that the solution of this RH problem is explicitly given by:
\begin{equation}\label{lol 1}
P^{(0)}(z) =  P_{0}^{(\infty)}(z) \left( \frac{z+x}{z} \right)^{\sigma_{3}} \begin{pmatrix}
1 & f(z;x) \\ 0 & 1
\end{pmatrix} z^{\frac{\alpha + 2}{2}\sigma_{3}} \begin{pmatrix}
1 & h(z) \\ 0 & 1
\end{pmatrix} H_{0}(z)
\end{equation}
where
\begin{equation}
f(z;x) = \frac{-z^2}{2\pi i} \int_{-ax}^{0} \frac{|s|^{\alpha}}{s-z}ds.
\end{equation}

\subsubsection{Small norm RH problem}

Define
\begin{equation}\label{SmallNormXto0}
R(z) = \left\{ \begin{array}{l l}
\widetilde{W}(z)P^{(\infty)}(z)^{-1}, & \mbox{for } z \in \mathbb{C}\setminus \overline{D_{0}}, \\[0.1cm]
\widetilde{W}(z)P^{(0)}(z)^{-1}, & \mbox{for } z \in D_{0}.
\end{array} \right.
\end{equation}
By definition of $\widetilde{W}$, $P^{(\infty)}$ and $P^{(0)}$, $R$ is analytic on $\mathbb{C}\setminus\partial D_{0}$. Let us put the clockwise orientation on $\partial D_{0}$. The jumps of $R$ on $\partial D_{0}$ are given by 
\begin{equation}
R_{-}(z)^{-1}R_{+}(z) = P^{(0)}(z)P^{(\infty)}(z)^{-1} = I + \bigO(ax), \qquad \mbox{ as } ax \to 0,
\end{equation}
where we have used \eqref{Matching Cond P0}. Also, from \eqref{Asympt tilde Phi z to inf} and \eqref{Asymp x small Pinf z to inf}, as $z \to \infty$ we have $R(z) = I + \bigO(z^{-1})$. Thus, by standard theory for small norm RH problems, $R$ exists for sufficiently small $ax$ and satisfies
\begin{equation}\label{lol38}
R(z) = I + \bigO(ax), \qquad R^{\prime}(z) = \bigO(ax),
\end{equation}
as $ax \to 0$ uniformly in $z \in \mathbb{C}\setminus \partial D_{0}$. Also, as $z \to \infty$, we have
\begin{equation}\label{lol8}
R(z) = I + \frac{R_{1}(x;a)}{z} + \bigO(z^{-2}),
\end{equation}
where $R_{1}(x;a) = \bigO(ax)$ as $ax \to 0$. For $z \in \mathbb{C}\setminus \overline{D_{0}}$, we have $\widetilde{W}(z) = R(z) P^{(\infty)}(z)$ and therefore we obtain
\begin{align}
& \Phi_{1}(x;a) = P_{1}^{(\infty)} + R_{1}(x;a). \label{lol32}
\end{align}
In particular, from \eqref{lol30} and \eqref{lol32}, as $ax \to 0$ we obtain
\begin{equation}\label{v asymptotics for x to 0}
v(x;a) = v(0) + \bigO(ax),
\end{equation}
with $v(0) := \frac{1}{8}(4(\alpha + 2)^{2}-1)$. To obtain asymptotics for $b_{1}$ and $b_{2}$, we will proceed similarly as done in Subsection \ref{Small norm RH problem for x small}, but instead of the global parametrix, we will need the local parametrix. From \eqref{A equation in the lax pair}, \eqref{para A} and \eqref{b in the lax pair}, we have
\begin{align}
& b_{1}(\sqrt{x}) = \lim_{z \to -1} i \sqrt{x} (z+1) \left[ \partial_{z}\Phi(xz;x,a)\Phi^{-1}(xz;x,a) \right]_{12}, \label{b1 lol} \\
& b_{2}(\sqrt{x}) = \lim_{z \to -a} i \sqrt{x} (z+a) \left[ \partial_{z}\Phi(xz;x,a)\Phi^{-1}(xz;x,a) \right]_{12}.
\end{align}
On the other hand, from \eqref{W tilde def}, \eqref{lol 1} and \eqref{SmallNormXto0}, we have as $ax \to 0$ and for $z \in D_{0}$
\begin{equation}\label{lol7}
\hspace{-0.6cm}\Phi(xz;x,a) = R(xz)P_{0}^{(\infty)}(xz) \left( \frac{z+1}{z} \right)^{\hspace{-0.06cm}\sigma_{3}} \hspace{-0.15cm} \begin{pmatrix}
1 & f(xz;x) \\ 0 & 1
\end{pmatrix} (xz)^{\frac{\alpha + 2}{2}\sigma_{3}} \hspace{-0.06cm} \begin{pmatrix}
1 & h(xz) \\ 0 & 1
\end{pmatrix} \hspace{-0.06cm} H_{-ax}(z).
\end{equation}
Thus, by the estimate \eqref{lol38} and a direct calculation, we have
\begin{equation}\label{lol 2}
b_{2}(\sqrt{x}) = i \sqrt{x} (1+\bigO(ax)) \left[ P_{0}^{(\infty)}(0) \begin{pmatrix}
0 & \star \\ 0 & 0
\end{pmatrix} P_{0}^{(\infty)}(0)^{-1} \right]_{12}, \qquad \mbox{ as } ax \to 0,
\end{equation}
where in the above equation
\begin{align}
& \star = \lim_{z \to -a} (z+a) \frac{(z+1)^{2}}{z^{2}} \partial_{z}f(xz;x) \\
& \hspace{0.2cm} = \frac{-(a-1)^{2}x^{2}}{2\pi i} \lim_{z\to -a} (z+a) \int_{-ax}^{0} \frac{x|s|^{\alpha}}{(s-xz)^{2}}ds = \frac{(a-1)^{2}x^{2}}{2\pi i}(ax)^{\alpha}.
\end{align}
Therefore, \eqref{lol 2} becomes
\begin{equation}\label{lol6}
b_{2}(\sqrt{x}) = \frac{(1-a^{-1})^{2}}{2\pi}\sqrt{x} (ax)^{\alpha+2} P_{0,11}^{(\infty)}(0)^{2}(1+\bigO(ax)), \qquad \mbox{ as } ax \to 0.
\end{equation}
To compute $P_{0,11}^{(\infty)}(0)$, we can use \eqref{lol3} and \eqref{lol4}. For $z \in \{ z \in \mathbb{C}: |\arg(z)|<\tfrac{2\pi}{3} \}$ we have
\begin{equation}\label{lol5}
P_{0,11}^{(\infty)}(z) = \sqrt{\pi} I_{\alpha+2}(\sqrt{z}) z^{- \tfrac{\alpha+2}{2}}.
\end{equation}
Taking the limit $z \to 0$ in \eqref{lol5}, and using the small $z$ expansion of $I_{\alpha+2}(z)$ (see \cite[formula 10.30.1]{NIST}) we obtain $P_{0,11}^{(\infty)}(0) = \frac{\sqrt{\pi}}{2^{\alpha+2}\Gamma(\alpha +3)}$. Inserting this value in \eqref{lol6}, we have as $ax \to 0$ that
\begin{align}
& b_{2}(\sqrt{x}) = \frac{(1-a^{-1})^{2}\sqrt{x} (ax)^{\alpha+2}}{2^{2\alpha+5}\Gamma(\alpha+3)^{2}}(1+\bigO(ax)) \\
& \hspace{1.2cm} = \frac{(1-a^{-1})^{2} \sqrt{x}}{2}I_{\alpha+2}^{2}(\sqrt{ax}))(1+\bigO(ax)). \label{lol35}
\end{align}
Similarly, from \eqref{b1 lol} and \eqref{lol7}, we have as $ax \to 0$ that
\begin{equation}\label{lol34}
b_{1}(\sqrt{x}) = i \sqrt{x} \left[ P_{0}^{(\infty)}(0)\sigma_{3}P_{0}^{(\infty)}(0)^{-1} \right]_{12}(1+\bigO(ax)) = -2i \sqrt{x} P_{0,11}^{(\infty)}(0)P_{0,12}^{(\infty)}(0)(1+\bigO(ax)).
\end{equation}
Again, from \eqref{lol4} and \eqref{lol3}, we obtain for $z \in \{ z \in \mathbb{C}: |\arg(z)|<\tfrac{2\pi}{3} \}$ that
\begin{equation}
P_{0,12}^{(\infty)}(z) = \frac{i}{\sqrt{\pi}}K_{\alpha+2}(\sqrt{z})z^{\frac{\alpha+2}{2}}-P_{0,11}(z)h(z)z^{\alpha+2}.
\end{equation}
By taking the limit $z \to 0$ and using \cite[formulas 10.30.2]{NIST}, this gives $P_{0,12}^{(\infty)}(0)= \frac{i}{\sqrt{\pi}}  2^{\alpha+1} \Gamma(\alpha+2)$, and by \eqref{lol34} we have
\begin{equation}
b_{1}(\sqrt{x}) = \frac{\sqrt{x}}{\alpha+2}(1+\bigO(ax)), \qquad \mbox{ as } ax \to 0. \label{lol36}
\end{equation}
With the change of functions \eqref{change of functions}, we obtain from \eqref{lol35} and \eqref{lol36} the small $ax$ asymptotics for $q_{1}(x;a)$ and $q_{2}(x;a)$ given in \eqref{asymp q1 origin} and \eqref{asymp q2 origin}.
We will also need later the asymptotics of $\Phi(z;x,a)$ as $z \to \infty$ and simultaneously $ax \to 0$. This can be obtained from \eqref{W tilde def}, \eqref{Asymp x small Pinf z to inf}, \eqref{SmallNormXto0} and \eqref{lol8}, and by the fact that the global parametrix \eqref{lol9} is independent of $a$ and $x$, we have
\begin{align}
& \Phi(z;x,a) = \left( I+\bigO(\tfrac{ax}{z}) \right) (I+\bigO(z^{-1}))z^{-\frac{1}{4}\sigma_3}Ne^{z^\frac{1}{2}\sigma_3}, & & \mbox{ as } z \to \infty \mbox{ and } ax \to 0, \nonumber \\
& \hspace{1.54cm}= (I+\bigO(z^{-1}))z^{-\frac{1}{4}\sigma_3}Ne^{z^\frac{1}{2}\sigma_3}, & & \mbox{ as } z \to \infty \mbox{ and } ax \to 0. \label{lol15}
\end{align}
\section{Steepest descent analysis of $Y$ as $nyr \to 0$}
\label{Section Steepest descent Y}

An essential ingredient in the steepest descent analysis is the equilibrium measure $\mu$, which in our case is the unique probability measure which minimizes
\begin{equation}\label{min prob measures}
\int_{0}^{\infty}\int_{0}^{\infty} \log \frac{1}{|x-y|} d\tilde\mu(x)d\tilde\mu(y) + \int_{0}^{\infty} y d\tilde\mu(y),
\end{equation}
among all Borel probability measures $\tilde\mu$ on $(0,\infty)$. The unique solution $\mu$ of the minimization problem \eqref{min prob measures} is supported on $\overline{\mathcal{S}}$, where $\mathcal{S} := (0, 4)$, and its density is known as the Marchenko-Pastur law:
\begin{equation}\label{rho h R}
\frac{d\mu(x)}{dx} = \rho(x) = \frac{1}{2\pi} \sqrt{\frac{4-x}{x}}.
\end{equation}
The equilibrium measure $\mu$ satisfies the following identities \cite{SaTo}, known as the Euler-Lagrange variational conditions:
\begin{align}
&\label{var eq}2\int_{\mathcal{S}} \log|x-y|\rho(y)dy - x = \ell, \quad x\in \overline{\mathcal{S}}, \\
&\label{var ineq}2\int_{\mathcal{S}} \log|x-y|\rho(y)dy - x < \ell, \quad x\in (4,\infty),
\end{align}
where $\ell=-2$. We define the $g$-function by 
\begin{equation}
g(z) = \int_{\mathcal{S}}
\log(z-x) d\mu(x),
\end{equation}
where the principal branch of the logarithm is
taken, meaning $g$ is analytic on $\mathbb{C} \setminus (-\infty, 4]$. 
The $g$-function possesses the following properties
\begin{align}
g_+(x) + g_-(x) - x - \ell & = 0, & x & \in \mathcal{S} \label{g1},&\\
2g(x) - x - \ell & < 0, & x & \in (4,\infty),&\label{g2} \\
g_+(x) - g_-(x) & = 2\pi i \int^{4}_x \rho(s) ds, & x & \in \mathcal{S},& \label{g3} \\
g_+(x) - g_-(x) & = 2\pi i , & x & \in (-\infty,0).&
\end{align}
Let us also define
\begin{equation}\label{xi} 
\xi(z) = -\pi \int_{4}^{z} \tilde{\rho}(s)ds \qquad \text{for }z  \in \mathbb{C} \setminus (-\infty, 4],
\end{equation}
where the integration path does not cross $(-\infty, 4]$, and where $\tilde{\rho}(z) = \frac{1}{2\pi} \sqrt{\frac{z-4}{z}}$ is analytic in $\mathbb{C}\setminus \overline{\mathcal{S}}$ and such that $\tilde{\rho}_{\pm}(x) = \pm i \rho(x)$ for $x \in \mathcal{S}$.  By \eqref{g1} and \eqref{g3} we have, 
\begin{equation}\label{g xi 1}
2 \xi_{\pm}(x) = \pm(g_{+}(x)-g_{-}(x)) = 2g_{\pm}(x)-x-\ell, \qquad x \in \mathcal{S}.
\end{equation}
By analytically continuing $\xi-g$ on the whole complex plane from the above expression, we obtain the identity
\begin{equation}\label{g xi 2} 
2\xi(z)=2g(z)-z-\ell, \qquad z \in \mathbb{C}\setminus (-\infty,4].
\end{equation}
The jumps of $\xi$ follow from those of $g$, we have
\begin{align}
\xi_+(x) + \xi_-(x) & = 0, & x &\in \mathcal{S} \label{xi1},\\
2\xi(x) & < 0, & x & \in (4,\infty),\label{xi2} \\
\xi_+(x) - \xi_-(x) & = 2\pi i \int^{4}_x \rho(s) ds, & x & \in \mathcal{S}, \label{xi3} \\
\xi_+(x) - \xi_-(x) & = 2\pi i, & x & \in (-\infty,0). \label{xi4}
\end{align}

\subsection{Transformation to constant jumps}

The weight \eqref{weight} is defined on $(yr,\infty)$. We consider its natural extension
\begin{equation}\label{weight complex plane}
w(z) = (z-y)^{2}z^{\alpha}e^{-nz}, \qquad z \in \mathbb{C}\setminus (-\infty,0],
\end{equation}
where the principal branch is taken for the root, and we define $\Psi(z) = Y(z) w(z)^\frac{\sigma_3}{2}$. The matrix function $\Psi$ satisfies the following RH problem.
\subsubsection*{RH problem for $\Psi$}
\begin{itemize}
\item[(a)] $\Psi: \mathbb{C} \setminus ((-\infty,0] \cup \{y\} \cup [yr, \infty))  \rightarrow \mathbb{C}^{2 \times 2}$ is analytic.
\item[(b)] Let $j_\Psi(z) := \Psi_-(z)^{-1}\Psi_+(z)$. Then,
\begin{align}
j_\Psi(z) &= \begin{pmatrix}
1 & 1 \\
 0 & 1 \\
\end{pmatrix}, &z\in (yr, \infty).\\
j_\Psi(z) &= e^{\pi i \alpha \sigma_3}, &z\in (-\infty,0).
\end{align}
\item[(c)] $\Psi(z) = (I+\bigO(z^{-1}))z^{(n + \frac{\alpha+2}{2})\sigma_3}e^{-\frac{nz}{2}\sigma_{3}}$ as $z \rightarrow \infty$.
\item[(d)] $\Psi$ has the following behaviour near $0$, $y$ and $yr$:
\begin{align}
& \Psi(z) = \bigO(1) z^{\frac{\alpha}{2}\sigma_3}, & & \mbox{ as } z \to 0, \\[0.26cm]
& \Psi(z) = \bigO(1) (z-y)^{\sigma_3}, & & \mbox{ as } z \to y, \\[0.02cm]
& \Psi(z) = \bigO(1) \begin{pmatrix}
1 & -\frac{\log(yr-z)}{2 \pi i} \\ 0 & 1
\end{pmatrix}, & & \mbox{ as } z \to yr,
\end{align}
where in the three above asymptotics it can be verified that the $\bigO(1)$ terms are analytic in a neighbourhood of their respective point.
\end{itemize}

\subsection{Opening of the lenses}
We now perform the step of opening the lenses. Since the jumps for $\Psi$ are constant, the lens contours are unconstrained. In a subsequent transformation we will use the $g$-function to normalise the RH problem at infinity, at which point the lens contours will be required to stay within a region in which they converge to the identity matrix as $n \rightarrow \infty$. Let $D_{0}$ and $D_{4}$ denote small but fixed open discs centred at $0$ and $4$ respectively and $U = D_{0} \cup D_{4}$. Note that since $nyr \to 0$ as $n \to \infty$, the points $0$, $y$ and $yr$ lie in $D_0$ for sufficiently large $n$. Let us define $\gamma = (yr, 4)$. We now make the transformation 
\begin{equation}\label{S Psi transformation}
S(z) := \Psi(z) K(z)
\end{equation}
where $K$ is a piecewise function designed to open the lens,
\beq
K(z):= \left\{
\begin{array}{lr}
I,  & \mbox{for } z \in \mathbb{C} \setminus (\Omega^{\gamma}_+ \cup \Omega^{\gamma}_-), \\
\begin{pmatrix}1&0\\-1&1\end{pmatrix}, & \mbox{for } z \in \Omega^{\gamma}_+,\\
\begin{pmatrix}1&0\\1&1\end{pmatrix}, & \mbox{for } z \in \Omega^{\gamma}_-.\\
\end{array}
\right.
\eeq
The regions $\Omega^{\gamma}_+$ and $\Omega^{\gamma}_-$ are shown in Figure \ref{open lens contours}, as well as their boundaries $\partial \Omega_{+}^{\gamma} = \Sigma_{+}^{\gamma} \cup \gamma$ and $\partial \Omega_{-}^{\gamma} = \Sigma_{-}^{\gamma} \cup \gamma$. 
\begin{figure}[t]
    \begin{center}
    \setlength{\unitlength}{1truemm}
    \begin{picture}(100,55)(28,10)
        \put(58,40){\line(1,0){70}}
        \put(50,40){\line(-1,0){30}}
        \put(85,39.9){\thicklines\vector(1,0){.0001}}
        \put(82,44.5){$\Omega_{+}^{\gamma}$}
        \put(82,33){$\Omega_{-}^{\gamma}$}
        \put(120,39.9){\thicklines\vector(1,0){.0001}}
        \put(35,39.9){\thicklines\vector(1,0){.0001}}
        \put(50,40){\thicklines\circle*{1.2}}
        \put(49,36.8){$0$}
        \put(54,40){\thicklines\circle*{1.2}}
        \put(52,36.8){$y$}
        \put(58,40){\thicklines\circle*{1.2}}
        \put(57.2,36.4){$yr$}
        \put(110,40){\thicklines\circle*{1.2}}
        \put(109,36.4){$4$}
        \qbezier(58,40)(84,68)(110,40)
        \put(85,54){\thicklines\vector(1,0){.0001}}
        \put(82,55.5){$\Sigma_{+}^{\gamma}$}
        \qbezier(58,40)(84,12)(110,40)
        \put(85,26){\thicklines\vector(1,0){.0001}}
        \put(82,22){$\Sigma_{-}^{\gamma}$}
    \end{picture}
    \caption{Jump contours for the RH problem for $S$. The lens contours are labelled $\Sigma^\gamma_+$ and $\Sigma^\gamma_-$ while the upper and lower lens regions are labelled $\Omega^\gamma_+$ and $\Omega^\gamma_-$. \label{open lens contours}}
\end{center}
\end{figure}
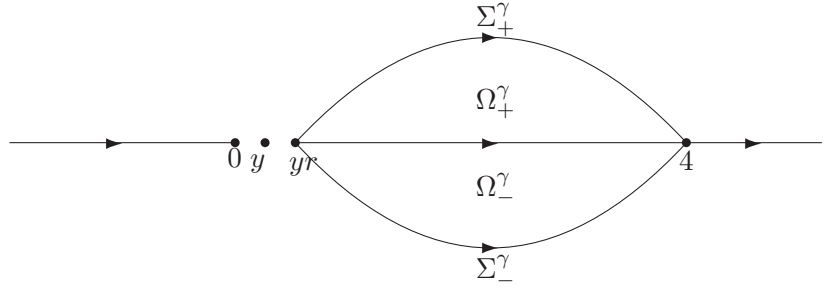
The function $S$ satisfies the following RH problem.
\subsubsection*{RH problem for $S$}
\begin{itemize}
\item[(a)] $S: \mathbb{C} \setminus ((-\infty,0] \cup \{y \} \cup [yr, \infty) \cup \Sigma^\gamma) \rightarrow \mathbb{C}^{2 \times 2}$ is analytic, where $\Sigma^{\gamma} = \Sigma_{+}^{\gamma}\cup\Sigma_{-}^{\gamma}$ is shown in Figure \ref{open lens contours}.
\item[(b)] Let $j_S(z) := S_-(z)^{-1}S_+(z)$. We have,
\begin{align}
j_S(z) &= \begin{pmatrix}
 0 & 1 \\
 -1 & 0 \\
\end{pmatrix}, &z\in \gamma,\\
j_S(z) &= \begin{pmatrix}
1 & 1\\
 0 & 1 \\
\end{pmatrix}, &z\in (4, \infty),\\
j_S(z) &= e^{\pi i \alpha \sigma_3}, &z\in (-\infty,0),\\
j_S(z) &= \begin{pmatrix}
 1 & 0 \\
 1 & 1 \\
\end{pmatrix}, &z \in \Sigma^\gamma.
\end{align}
\item[(c)] $S(z) = (I+\bigO(z^{-1}))z^{(n + \frac{\alpha +2}{2})\sigma_3}e^{-\frac{nz}{2}\sigma_{3}}$ as $z \rightarrow \infty$.
\item[(d)] Near $0$, $y$, $yr$ and $4$, the behaviour of $S$ takes the form
\begin{equation}
\begin{array}{l l}
\displaystyle S(z) = \bigO(1) z^{\frac{\alpha}{2}\sigma_3}, & \displaystyle \mbox{as } z \to 0, \\[0.25cm]
\displaystyle S(z) = \bigO(1) (z-y)^{\sigma_3}, & \displaystyle \mbox{as } z \to y, \\[0.1cm]
\displaystyle S(z) = \bigO(1) \begin{pmatrix}
1 & -\frac{\log(yr-z)}{2 \pi i} \\ 0 & 1
\end{pmatrix} K(z), & \displaystyle \mbox{as } z \to yr, \\[0.1cm]
\displaystyle S(z) = \bigO(1)K(z), & \displaystyle \mbox{as } z \to 4,
\end{array}
\end{equation}
where in the above asymptotics the $\bigO$ terms are analytic in a neighbourhood of their respective point.
\end{itemize}

\subsection{Normalisation at infinity}

The next transformation takes the form,
\begin{equation}\label{T S Transformation}
T(z) = e^{-\frac{n \ell \sigma_3}{2}} S(z) \times \begin{cases} e^{-n\xi(z) \sigma_3},&\mbox{if $z \in \mathbb{C} \setminus \overline{U},$}\\
I  ,&\mbox{if $z \in U.$ }\end{cases}
\end{equation}
The above transformation has the effect of normalising the problem at infinity.

\subsubsection*{RH problem for $T$}
\begin{itemize}
\item[(a)] $T: \mathbb{C} \setminus ((-\infty,0] \cup \{y \} \cup [yr, \infty) \cup \Sigma^\gamma \cup \partial U) \rightarrow \mathbb{C}^{2 \times 2}$ is analytic.
\item[(b)] Let $j_T(z) := T_-(z)^{-1}T_+(z)$. Then,
\begin{align*}
j_T(z) &= j_S(z), & z &\in ((-\infty,0) \cup (yr, \infty) \cup \Sigma^\gamma) \cap U,\\ 
j_T(z) &= \begin{pmatrix}
 0 & 1 \\
 -1 & 0 \\
\end{pmatrix}, &z &\in \gamma \setminus \overline{U},\\
j_T(z) &= \begin{pmatrix}
1&e^{2n \xi(z)} \\
 0 & 1 \\
\end{pmatrix}, &z&\in (4, \infty) \setminus \overline{U},\\
j_T(z) &= \begin{pmatrix}
1 & 0 \\
e^{-2 n \xi(z)} & 1 \\
\end{pmatrix}, &z &\in \Sigma^\gamma \setminus \overline{U},\\
j_T(z) &= e^{\pi i \alpha \sigma_3}, &z&\in (-\infty,0),\\
j_T(z) &= e^{-n  \xi(z) \sigma_3}, & z &\in \partial U,
\end{align*}
where the orientation of $\partial U$ is clockwise.
\item[(c)] As $z \rightarrow \infty$,
\beq
T(z) = (1+ \bigO(z^{-1})) z^{\frac{\alpha+2}{2} \sigma_3}.
\eeq
\item[(d)] Near the endpoints, the behaviour of $T$ takes the form
\begin{equation}
\begin{array}{l l}
\displaystyle T(z) = \bigO(1) z^{\frac{\alpha}{2}\sigma_3}, & \displaystyle \mbox{as } z \to 0, \\[0.25cm]
\displaystyle T(z) = \bigO(1) (z-y)^{\sigma_3}, & \displaystyle \mbox{as } z \to y, \\[0.1cm]
\displaystyle T(z) = \bigO(1) \begin{pmatrix}
1 & -\frac{\log(yr-z)}{2 \pi i} \\ 0 & 1
\end{pmatrix} K(z), & \displaystyle \mbox{as } z \to yr, \\[0.1cm]
\displaystyle T(z) = \bigO(1), & \displaystyle \mbox{as } z \to 4.
\end{array}
\end{equation}
\end{itemize}

\subsection{Global parametrix}

Finally we define the global parametrix as a function $P^{(\infty)}$ satisfying the following RH problem.
\subsubsection*{RH problem for $P^{(\infty)}$}
\begin{itemize}
\item[(a)] $P^{(\infty)}: \mathbb{C} \setminus (-\infty,4] \rightarrow \mathbb{C}^{2 \times 2}$ is analytic.
\item[(b)] $P^{(\infty)}$ has the jump relations
\begin{align}
P^{(\infty)}_+(z) &= P^{(\infty)}_-(z) \begin{pmatrix}0 &1 \\ -1 &0\end{pmatrix},\qquad  z \in \mathcal{S} \\
P^{(\infty)}_+(z) &= P^{(\infty)}_-(z) e^{\pi i \alpha \sigma_3}, \qquad z \in (-\infty,0).
\end{align}
\item[(c)] As $z \rightarrow \infty$,
\beq
P^{(\infty)}(z) = (I + \bigO(z^{-1}))z^{\frac{\alpha + 2}{2}\sigma_3}.
\eeq
\item[(d)] As $z \rightarrow \hat{z} \in \partial \mathcal{S} = \{0,4\}$, we have
\beq
P^{(\infty)}(z) = \bigO((z-\hat{z})^{-\frac{1}{4}}).
\eeq
\end{itemize}
The solution is explicitly given by
\begin{equation}
\label{pinf copy paste}
P^{(\infty)}(z) =  N^{-1} \left(\frac{z-4}{z} \right)^{-\frac{\sigma_3}{4}} N \varphi\left(\frac{z}{2}-1\right)^{\frac{\alpha+2}{2}\sigma_3},
\end{equation}
where $\varphi(z) = z + \sqrt{z^2 - 1}$ is analytic in $\mathbb{C}\setminus [-1,1]$.
We now need to construct local parametricies valid in the fixed open discs $D_{0}$ and $D_{4}$ around $0$ and $4$. 

\subsection{Local parametrix near $4$}

\subsubsection*{RH problem for $P^{(4)}$}
\begin{itemize}
\item[(a)] $P^{(4)}: D_{4}\setminus (\Sigma^{\gamma} \cup \mathbb{R}) \rightarrow \mathbb{C}^{2 \times 2}$ is analytic.
\item[(b)] $P^{(4)}$ has the jump relations
\begin{equation}\label{jumps P3}
\begin{array}{l l}
P^{(4)}_{+}(z) = P^{(4)}_{-}(z) \begin{pmatrix}
0 & 1 \\ -1 & 0
\end{pmatrix}, & \mbox{ on } (-\infty,4)\cap D_{4}, \\

P^{(4)}_{+}(z) = P^{(4)}_{-}(z) \begin{pmatrix}
 1 & 1 \\
 0 & 1
\end{pmatrix}, & \mbox{ on } (4,\infty)\cap D_{4}, \\

P^{(4)}_{+}(z) = P^{(4)}_{-}(z) \begin{pmatrix}
 1 & 0  \\ 1 & 1
\end{pmatrix}, & \mbox{ on } \Sigma^{\gamma}\cap D_{4}.
\end{array}
\end{equation}
\item[(c)] As $z \rightarrow 4$, $P^{(4)}(z) = \bigO(1)$.
\item[(d)] As $n \to \infty$, $P^{(4)}(z) = \left(I+ \bigO(n^{-1})\right) P^{(\infty)}(z) e^{n\xi(z)\sigma_{3}}$ uniformly for $z \in \partial D_{4}$.
\end{itemize}
The solution $P^{(4)}$ can be constructed in term of the solution $\widetilde{\Upsilon}$ of the Airy model RH problem parametrix, which is presented in the appendix, see Subsection \ref{Subsection: Airy model RH problem}. The local parametrix inside $D_{4}$ can then be written as
\begin{equation}\label{P at 4}
P^{(4)}(z) = \widetilde{E}(z) \widetilde\Upsilon(n^{2/3}\widetilde{f}(z)),
\end{equation}
where $\widetilde{f}(z)$ is given by 
\begin{equation}
\widetilde{f}(z) := \left( -\frac{3}{2} \xi(z) \right)^{2/3}.
\end{equation}
From the definition of $\xi$ given by \eqref{xi}, $\widetilde{f}$ is a conformal map from a neighbourhood of $4$ to a neighbourhood of $0$. The matrix function $\widetilde{E}$ is defined in $D_{4}$ by
\begin{equation}
\widetilde{E}(z) = P^{(\infty)}(z) N^{-1} \widetilde{f}(z)^{\frac{\sigma_{3}}{4}}n^{\frac{\sigma_{3}}{6}},
\end{equation}
where the principal branch is taken for $(\cdot)^{\frac{1}{4}}$. It can be directly verified from the RH problem for $P^{(\infty)}$ and the definition of $\widetilde{f}$ that $\widetilde{E}$ is analytic in $D_{4}$, and from the properties of $\widetilde{\Upsilon}$ (presented in Subsection \ref{Subsection: Airy model RH problem}) that $P^{(4)}$ given by \eqref{P at 4} satisfies indeed the above RH problem.

\subsection{Local parametrix near $0$}

Inside $D_0$ we require a local parametrix satisfying the following RH problem.
\subsubsection*{RH problem for $P^{(0)}$}
\label{local parametrix}
\begin{itemize}
\item[(a)] $P^{(0)} : D_{0} \setminus ((-\infty,0] \cup \{y \} \cup [yr, \infty) \cup \Sigma^\gamma) \rightarrow \mathbb{C}^{2 \times 2}$ is analytic.
\item[(b)] $P^{(0)}$ has the jump relations
\begin{equation}\label{jumps P3}
\begin{array}{l l}
P^{(0)}_{+}(z) = P^{(0)}_{-}(z) \begin{pmatrix}
0 & 1 \\ -1 & 0
\end{pmatrix}, & \mbox{ on } (yr,\infty)\cap D_{0}, \\

P^{(0)}_{+}(z) = P^{(0)}_{-}(z) e^{\pi i \alpha \sigma_{3}}, & \mbox{ on } (-\infty,0)\cap D_{0}, \\

P^{(0)}_{+}(z) = P^{(0)}_{-}(z) \begin{pmatrix}
 1 & 0  \\ 1 & 1
\end{pmatrix}, & \mbox{ on } \Sigma^{\gamma}\cap D_{0}.
\end{array}
\end{equation}
\item[(c)] Near $0$, $y$ and $yr$, the behaviour of $P^{(0)}$ takes the form
\begin{equation}
\begin{array}{l l}
\displaystyle P^{(0)}(z) = \bigO(1) z^{\frac{\alpha}{2}\sigma_3}, & \displaystyle \mbox{as } z \to 0, \\[0.25cm]
\displaystyle P^{(0)}(z) = \bigO(1) (z-y)^{\sigma_3}, & \displaystyle \mbox{as } z \to y, \\[0.1cm]
\displaystyle P^{(0)}(z) = \bigO(1) \begin{pmatrix}
1 & -\frac{\log(yr-z)}{2 \pi i} \\ 0 & 1
\end{pmatrix} K(z), & \displaystyle \mbox{as } z \to yr, \\[0.1cm]
\end{array}
\end{equation}
where in the above asymptotics the $\bigO$ terms are analytic in a neighbourhood of their respective point.
\item[(d)] As $n\rightarrow \infty$, $P^{(0)}(z) = (I + o(1)) P^{(\infty)}(z) e^{n \xi(z)\sigma_3}$ uniformly for $z \in \partial D_{0}$.
\end{itemize}
The solution $P^{(0)}$ uses the model RH problem $\Phi$ presented in Section \ref{section RH model Phi}, with the parameters $x$ and $a$ chosen such that
\begin{equation}\label{relations a r and x y}
x = n^{2}f(y) \qquad \mbox{ and } \qquad a = \frac{f(yr)}{f(y)},
\end{equation}
where $f$ is given by
\begin{equation}\label{conformal map at 0}
f(z) = -(\xi(z) - \xi(0))^2.
\end{equation}
From the definition of $\xi$ given by \eqref{xi}, $f$ is a conformal map from a neighbourhood of $0$ to a neighbourhood of $0$, and satisfies $f^{\prime}(0) = 4$. Note that since $nyr \to 0$ as $n \to \infty$, this implies that $x = 4n^{2}y(1+\bigO(y))$ and $a = r(1+\bigO(yr))$ as $n \to \infty$.
\begin{lemma}\label{lemma: local param}
As $n \to \infty$ and simultaneously $nyr\to 0$, the matrix function
\begin{equation} \label{Px*}
P^{(0)}(z) = E(z)  \sigma_3 \Phi(-n^2 f(z); n^2 f(y), f(yr)/f(y)) \sigma_3 e^{\frac{\pi i \alpha }{2}\theta(z) \sigma_3}
\end{equation}
satisfies the RH problem for $P^{(0)}$, where $E$ is the analytic function in $D_{0}$ given by
\begin{equation}
E(z) = (-1)^n P^{(\infty)}(z) e^{-\frac{\pi i \alpha }{2}\theta(z)\sigma_{3}} N (-f(z))^\frac{\sigma_3}{4}n^\frac{\sigma_3}{2},
\end{equation}
the function $f$ is given by \eqref{conformal map at 0}, and $\Phi(z;x,a)$ is the model RH problem introduced in Section \ref{PhiRH}. The $o(1)$ in the condition (d) of the RH problem for $P^{(0)}$ can be specified as
\begin{equation}\label{small o of 1}
\bigO(\max\{n^{-1},nry\}) = \left\{ \begin{array}{l l}
\bigO(n^{-1}), & \mbox{if } (n^{2}y,r) \mbox{ are in a compact subset of } (0,\infty)\times (1,\infty), \\
\bigO(n^{-1}), & \mbox{if } n^{2} r y \to 0, \\
\bigO\left( nry \right), & \mbox{if } nry \to 0 \mbox{ and } (r-1)n^{2}y \to \infty.
\end{array} \right.
\end{equation}
\end{lemma}
\begin{proof}
The analyticity of $E$ inside $D_{0}$ follows from the RH problem for $P^{(\infty)}$ and the definition of $f$. By definition of $\Phi$ (see Section \ref{section RH model Phi}), $P^{(0)}$ satisfies the RH problem for $P^{0}$. The explicit forms of $o(1)$ given by \eqref{small o of 1} follow from \eqref{Psic} (for $(x,a)$ in a compact subset of $(0,\infty)\times (1,\infty)$), \eqref{lol15} (for $ax \to 0$) and \eqref{lol16} (for $(a-1)x \to \infty$ and $\frac{ax}{n} \to 0$).
\end{proof}
Finally, we define $R(z)$ as follows
\begin{equation}\label{R final Y}
R(z)=\begin{cases}
T(z)P^{(\infty)}(z)^{-1},&\mbox{ for $z \in \mathbb{C} \setminus \overline{U}$, }\\
T(z)P^{(0)}(z)^{-1},&\mbox{ for $z\in D_{0}$,} \\
T(z)P^{(4)}(z)^{-1},&\mbox{ for $z\in D_{4}$.} \\
\end{cases}
\end{equation}
Using the above definition we can derive the following RH problem for $R$.
\subsubsection*{RH problem for $R$}

\begin{itemize}
\item[(a)] $R: \mathbb{C} \setminus \Sigma_{R}\rightarrow \mathbb{C}^{2 \times 2} $ is analytic, where $\Sigma_{R} = ( (4,\infty) \cup \Sigma^\gamma \cup  \partial U ) \setminus U$.

\item[(b)] Let $j_R(z) := R_-(z)^{-1}R_+(z)$. We have
\begin{align}
\label{jumps of R: lens}
j_R(z) &= P^{(\infty)}(z) \begin{pmatrix} 1 & 0 \\
 e^{-2 n \xi(z)} & 1 \end{pmatrix} P^{(\infty)}(z)^{-1}, &z \in \Sigma^\gamma \setminus \overline{U},\\ 
\label{jumps of R: gamma}
j_R(z) &=  P^{(\infty)}(z) \begin{pmatrix}  1 & e^{2n\xi(z) }  \\ 0 & 1\end{pmatrix} P^{(\infty)}(z)^{-1}, & z \in (4,\infty) \setminus \overline{D_{4}},\\ 
\label{jumps of R: disc 4}
j_R(z) &= P^{(4)}(z) \begin{pmatrix} e^{-n \xi(z)} & 0 \\
 0 & e^{n \xi(z)} \end{pmatrix} P^{(\infty)}(z)^{-1}, &z\in\partial D_{4}, \\
\label{jumps of R: disc 0}
j_R(z) &= P^{(0)}(z) \begin{pmatrix} e^{-n \xi(z)} & 0 \\
 0 & e^{n \xi(z)} \end{pmatrix} P^{(\infty)}(z)^{-1}, &z\in\partial D_{0}.
\end{align}
\item[(c)] As

\vspace{-1.57cm}\begin{equation}\label{R asymp inf Y}
\hspace{-0.4cm} z \to \infty, \mbox{ we have }R(z) = I + \frac{R_1}{z} + \bigO(z^{-2}).
\end{equation}
\item[(d)] As $z \to b \in \{0,y,yr,4\}$, $R(z) = \bigO(1)$.
\end{itemize}
From \eqref{pinf copy paste}, \eqref{P at 4} and \eqref{Px*}, as $n \to \infty$ we have
\begin{align}
& j_R(z) = I + \bigO(\max\{n^{-1},nry\}), & & \mbox{ uniformly for } z \in \partial D_{0}, \label{jumps R estimate 1} \\
& j_R(z) = I + \bigO(n^{-1}), & & \mbox{ uniformly for } z \in \partial D_{4}, \label{jumps R estimate 1 bis} \\
& j_R(z) = I + \bigO(e^{-cn}), & & \mbox{ uniformly for } z \in \Sigma_{R} \setminus \partial U. \label{jumps R estimate 2}
\end{align}
where $c > 0$ is a constant independent of $n$. It follows from standard theory for small norm RH problem that $R$ exists for $n$ sufficiently large and satisfies
\begin{equation}\label{R = I + small}
R(z) = I + \bigO(\max\{n^{-1},nry\}) \qquad \mbox{ and } \qquad  \partial_{z}R(z) = \bigO(\max\{n^{-1},nry\}),
\end{equation}
uniformly for  $z$ in compact subsets of $\mathbb{C}\setminus \Sigma_{R}$.

\subsection{Computation of $R_1$}
The quantity $R_1$ can be computed via a perturbative calculation. From \eqref{jumps R estimate 1} and \eqref{jumps R estimate 1 bis} we can write for $z \in \Sigma_{R}$,
\begin{equation}\label{JR as2}
j_R(z)=I+\frac{1}{n}J_1(z)+\bigO(n^{-2}),\qquad n\to\infty,
\end{equation}
where the matrix $J_1(z)$ is non-zero only on $\partial U$, satisfies $J_{1}(z) = \bigO(1)$ uniformly for $z \in \partial D_{4}$ and 
\begin{equation}\label{J_1 estimate}
J_{1}(z) = \bigO(\max\{1,n^{2}ry\}), \qquad \mbox{uniformly for }z \in \partial D_{0}.
\end{equation}
Therefore, by a perturbative analysis of $R$, we have
\begin{equation}\label{R large n}
R(z)=I+R^{(1)}(z)n^{-1}+\bigO(\max\{n^{-2},(nry)^{2}\}),\qquad z\in \mathbb{C} \setminus \Sigma_{R},
\end{equation}
where
\begin{equation}
R^{(1)}(z) = \bigO(\max\{1,n^{2}ry\}), \qquad \mbox{uniformly for }z \in \mathbb{C}\setminus \Sigma_{R}.
\end{equation}
The quantity $R^{(1)}(z)$ may be expressed in terms of $J_1(z)$ by substituting (\ref{JR as2}) into the jump relation $R_+(z)=R_-(z)j_R(z)$, from which we obtain the following RH problem for $R^{(1)}$.
\subsubsection*{RH problem for $R^{(1)}$}
\begin{itemize}
\item[(a)] $R^{(1)}:\mathbb C\setminus \partial U\to\mathbb C^{2\times 2}$ is analytic,
\item[(b)] $R_+^{(1)}(z)=R_-^{(1)}(z)+J_1(z)$ for $z\in \partial U$,
\item[(c)] $R^{(1)}(z)\to 0$ as $z\to\infty$. 
\end{itemize}
The above RH problem can be solved explicitly in terms of a Cauchy transform,
\begin{equation}\label{R^1 eqn}
R^{(1)}(z)=\frac{1}{2\pi i}\int_{\partial U}\frac{J_1(\xi)}{\xi-z}d\xi,
\end{equation}
where the integral is taken entry-wise. Explicit computations using \eqref{Psic}, \eqref{Px*} and \eqref{jumps of R: disc 0} gives for $z \in \partial D_{0}$
\begin{equation}\label{J_1 inside D_0}
J_1(z) = 
\frac{v(n^2 f(y);f(ry)/f(y))}{2 \sqrt{-f(z)}}P^{(\infty)}(z)\begin{pmatrix}
 -1 & -ie^{-i \pi  \alpha \theta(z)}  \\
 -ie^{i \pi  \alpha \theta(z) } & 1 \\
\end{pmatrix}P^{(\infty)}(z)^{-1}.
\end{equation}
The function $v$ is the special function appearing in the model problem $\Phi$. By using \eqref{P at 4}, \eqref{jumps of R: disc 4} and \eqref{Asymptotics Airy}, the term $J_{1}(z)$ on $\partial D_{4}$ is given by,
\begin{equation}\label{J_1 inside D_4}
J_1(z) = \frac{1}{8\widetilde{f}(z)^{3/2}} P^{(\infty)}(z)
\begin{pmatrix}
 \frac{1}{6} & i \\
 i & -\frac{1}{6} \\
\end{pmatrix} P^{(\infty)}(z)^{-1}.
\end{equation}
Putting \eqref{J_1 inside D_0} and \eqref{J_1 inside D_4} in \eqref{R^1 eqn} gives (after a residue calculation)
\begin{multline}
R^{(1)}(z) = \frac{v(n^2 f(y);f(ry)/f(y))}{2z} \begin{pmatrix}
1 & i \\ i & -1
\end{pmatrix} + \frac{5}{12(z-4)^{2}}\begin{pmatrix}
-1 & i \\ i & 1
\end{pmatrix} \\ + \frac{1}{16(z-4)}\begin{pmatrix}
1-4(\alpha+2)^{2} & \frac{i}{3}\left( 12(\alpha+2)^{2}+24(\alpha+2)+11 \right) \\
\frac{i}{3}\left( 12(\alpha+2)^{2}-24(\alpha+2)+11 \right) & 4(\alpha+2)^{2}-1
\end{pmatrix}.
\end{multline}
Therefore, we have
\begin{multline}
\hspace{-0.8cm}R_{1} = \lim_{z\to \infty} z(R(z)-I) = \frac{v(n^2 f(y);f(ry)/f(y))}{2n}\begin{pmatrix}
1 & i \\ i & -1
\end{pmatrix} + \\ \hspace{-0.8cm} \frac{1}{16n}\begin{pmatrix}
1-4(\alpha+2)^{2} & \frac{i}{3}\left( 12(\alpha+2)^{2}+24(\alpha+2)+11 \right) \\
\frac{i}{3}\left( 12(\alpha+2)^{2}-24(\alpha+2)+11 \right) & 4(\alpha+2)^{2}-1
\end{pmatrix} + \bigO(\max\{n^{-2},(nry)^{2}\}).
\end{multline}
In particular, one has
\begin{equation}\label{Trace R1}
n \mathrm{Tr}(R_1 \sigma_3) = v(n^2 f(y);f(ry)/f(y)) - v(0) + \bigO(\max\{n^{-1},n^{-1}(rn^{2}y)^{2}\}) \qquad \mbox{ as } n \to \infty,
\end{equation}
where $v(0) = \frac{1}{8}(4 (\alpha+2)^2 - 1)$.

\section{Proof of Theorem \ref{P asymptotics}} \label{proof1}


Let us define $s := 4n^{2}y$, which is a rescaling of $y$. In order to use Lemma \ref{diff identity}, we need to compute large $n$ asymptotics for $\mathrm{Tr}(Y^{-1}(z) \partial_{z}Y(z) \sigma_3)$ uniformly for $z$ in a neighbourhood of $\infty$. The large $n$ analysis for $Y$ done in Section \ref{Section Steepest descent Y} is valid when the parameters $y$ and $r$ satisfy fall in one of the three cases presented in \eqref{Case 1}, \eqref{Case 2} and \eqref{Case 3}. For large $z$, by \eqref{S Psi transformation}, \eqref{T S Transformation}, \eqref{pinf copy paste} and \eqref{R final Y}, we have
\begin{equation}
Y(z) = e^{\frac{n \ell \sigma_3}{2}} R(z) P^{(\infty)}(z) e^{n \xi(z) \sigma_3} w(z)^{-\frac{\sigma_3}{2}}.
\end{equation}
Using the above expression for $Y$, we obtain
\begin{equation*}
\begin{array}{r c l}
\displaystyle \mathrm{Tr}(Y^{-1}(z) \partial_{z}Y(z) \sigma_3) & = & \displaystyle -\partial_z \log w(z) + 2n \partial_z  \xi(z) \\[0.15cm]
& & \displaystyle + \mathrm{Tr}(P^{(\infty)}(z)^{-1}\partial_z P^{(\infty)}(z) \sigma_3) \\[0.15cm]
& & \displaystyle + \mathrm{Tr}(P^{(\infty)}(z)^{-1}R(z)^{-1} \partial_z R(z)P^{(\infty)}(z) \sigma_3).
\end{array}
\end{equation*}
Using \eqref{pinf copy paste}, as $z \to \infty$ we have
\begin{equation}
\mathrm{Tr}(P^{(\infty)}(z)^{-1}\partial_z P^{(\infty)}(z) \sigma_3) = \frac{\alpha+2}{z}  + \frac{2\alpha +4}{z^2}  + \bigO(z^{-3})
\end{equation}
Similarly, \eqref{xi} and \eqref{weight complex plane} give
\begin{equation}
-\partial_z \log w(z) + 2n \partial_z  \xi(z) = 2n\left(\frac{1}{z} + \frac{1}{z^2}\right) - \frac{\alpha+2}{z} - \frac{2y}{z^2} +\bigO(z^{-3}), \quad \mbox{ as } z \to \infty.
\end{equation}
Also, by \eqref{pinf copy paste} and \eqref{R asymp inf Y}, we have
\begin{equation}
\mathrm{Tr}(P^{(\infty)}(z)^{-1}R(z)^{-1} \partial_z R(z)P^{(\infty)}(z) \sigma_3) = -\frac{\mathrm{Tr}(R_{1}\sigma_{3})}{z^{2}}+\bigO(z^{-3}), \qquad \mbox{ as } z \to \infty.
\end{equation}
Using the above expressions and Lemma \ref{diff identity} gives,
\begin{align}
& z\mathrm{Tr}(Y^{-1}(z) \partial_{z}Y(z) \sigma_3) = 2n + \frac{2(\alpha+2)+2n-2y-\mathrm{Tr}(R_{1}\sigma_{3})}{z}+\bigO(z^{-2}), \quad \mbox{ as } z \to \infty, \\
& \label{partial y}
\partial_y \log Z_{n,\alpha}(y;r) = n + \frac{n}{2y} \mathrm{Tr}(R_1 \sigma_3).
\end{align}
By using \eqref{Trace R1} in \eqref{partial y} and rewriting it in terms of $s = 4n^{2}y$, we have as $n \to \infty$ that
\begin{equation}\label{eq 7.5}
\begin{array}{r c l}
\displaystyle \partial_{s} \log Z_{n,\alpha} \left( \frac{s}{4n^{2}};r \right)  & = & \displaystyle \frac{1}{4n} + \frac{n}{2s}\mathrm{Tr}(R_{1}\sigma_{3}), \\
& = & \displaystyle \frac{1}{2s} \left(  v\bigg(n^2 f(\tfrac{s}{4n^{2}}) ;  \frac{f(\frac{rs}{4n^{2}})}{f(\frac{s}{4n^{2}})} \bigg) - v(0) \right) + \frac{1}{s}\bigO(\max\{n^{-1},\tfrac{(rs)^{2}}{n}\}), \\
& = & \displaystyle \frac{1}{2s} \left(  v(s;r) - v(0) \right) + \frac{1}{s}\bigO(\max\{n^{-1},\tfrac{(rs)^{2}}{n}\}).
\end{array}
\end{equation}
Let us fix $r$. By integrating the left-hand side of \eqref{eq 7.5} from $\epsilon >0$ to a certain $s > \epsilon$ we obtain
\begin{equation}\label{useless}
\log\left( \frac{Z_{n,\alpha}(\tfrac{s}{4n^{2}};r)}{Z_{n,\alpha}(\tfrac{\epsilon}{4n^{2}};r)} \right) = \int_{\epsilon}^{s} \partial_{x} \log Z_{n,\alpha} \left( \frac{x}{4n^{2}};r \right)dx.
\end{equation}
Since the function $y \in [0,\infty) \mapsto Z_{n,\alpha}(y;r)$ is continuous, and since
\begin{equation}
Z_{n,\alpha}(0;r) = \widehat Z_{n,\alpha + 2} >0,
\end{equation} 
the left-hand-side of \eqref{useless} is bounded as $\epsilon \to 0$, and thus the same is true for the right-hand side. In order to use \eqref{eq 7.5} in \eqref{useless}, it is important to note that $\bigO$ term of \eqref{eq 7.5} is uniform when $s$ is in a compact subset of $(0,\infty)$ and also as $s \to 0$. Thus, we obtain
\begin{equation}
\label{small y Z}
\log \left( \frac{Z_{n,\alpha}(\frac{s}{4n^{2}};r)}{\widehat Z_{n,\alpha + 2}} \right) = \frac{1}{2} \int_0^{s} \left[v(x;r) - v(0)
 \right] \frac{dx}{x} + \bigO\big(\max\{n^{-1},\tfrac{(rs)^{2}}{n}\}\big).
\end{equation}
By an integration by parts, we have
\begin{equation*}
I(s;r) := \int_0^{s} \left[v(x;r) - v(0)
 \right] \frac{dx}{2x} = \frac{1}{2}\int_{0}^{s}v^{\prime}(x;r)\log\left( \frac{s}{x} \right)dx = -\frac{1}{4}\int_{0}^{s}\big(q_{1}^{2}(x;r)+rq_{2}^{2}(x;r)\big)\log\left( \frac{s}{x} \right)dx,
\end{equation*}
where for the last equality we have used \eqref{system for the b second} and \eqref{change of functions}. This completes the proof of Theorem \ref{P asymptotics}.
Note that \eqref{small y Z} can be rewritten as
\begin{equation}
Z_{n,\alpha}\Big(\frac{s}{4n^{2}};r\Big) = \widehat Z_{n,\alpha + 2} e^{I(s;r)} \left( 1 + \bigO\big(\max\{n^{-1},\tfrac{(rs)^{2}}{n}\}\big) \right).
\end{equation}

\section{Proof of Theorem \ref{thm-Q}} \label{proof2}

From \eqref{Qn} by changing variables $x := 4(n-1)ny$, we obtain
\begin{multline} \label{Qn in terms of x}
\hspace{-1.6cm}Q_{n,\alpha} (r) = \frac{\widehat Z_{n-1,\alpha+2}}{\widehat Z_{n,\alpha}}  \left( \frac{n-1}{n} \right)^{\hspace{-0.1cm}(n-1)(n+1+\alpha)}\hspace{-0.5cm} (4(n-1)n)^{-1-\alpha} \hspace{-0.1cm} \int_{0}^\infty \hspace{-0.15cm} \frac{x^\alpha e^{-\frac{x}{n}}}{\widehat Z_{n-1,\alpha+2}} Z_{n-1,\alpha}\hspace{-0.06cm}\left( \hspace{-0.1cm} \frac{x}{4(n-1)^{2}};r \hspace{-0.06cm} \right)\hspace{-0.06cm} dx.
\end{multline}
It is known that (see \cite[formula 17.6.5]{Mehta})
\begin{equation}
\displaystyle \widehat Z_{n,\alpha} = \frac{1}{n!}n^{-n^{2}-\alpha n} \prod_{j=1}^{n} j! \Gamma(j+\alpha),
\end{equation}
and therefore we have
\begin{equation}\label{Zrel}
\begin{array}{r c l}
\displaystyle \frac{\widehat Z_{n-1,\alpha+2}}{\widehat Z_{n,\alpha}} & = & \displaystyle \frac{(n-1)^{\alpha+1} \Gamma (n +\alpha + 1)}{(n-1)! \Gamma(\alpha + 1) \Gamma(\alpha + 2)} \left( \frac{n}{n-1} \right)^{n^{2} + \alpha n}, \\
& = & \displaystyle \frac{n^{2\alpha+2}}{\Gamma(\alpha+1)\Gamma(\alpha+2)} \left( \frac{n}{n-1} \right)^{n^{2} + \alpha n} \left( 1+\bigO\left(n^{-1}\right) \right), \quad \mbox{ as } n \to \infty.
\end{array}
\end{equation}
We can thus simplify expression \eqref{Qn in terms of x} as
\begin{equation}
Q_{n,\alpha}(r) = \frac{1+ \bigO(n^{-1})}{4^{\alpha+1}\Gamma(\alpha + 1) \Gamma(\alpha + 2)} (I_{1} + I_{2}) ,
\end{equation}
where
\begin{equation}
I_{1} = \int_{0}^{M} x^{\alpha} e^{ -\frac{x}{n}} e^{I(x;r)}dx, \qquad I_{2} = \int_{M}^{\infty}  \frac{x^\alpha e^{-\frac{x}{n}}}{\widehat{Z}_{n-1,\alpha+2}} Z_{n-1,\alpha}\left(\frac{x}{4(n-1)^{2}};r \right) dx,
\end{equation}
and $M > 0$ is a constant. The asymptotics \eqref{asymp I large} implies that for any $\epsilon >0$, we have
\begin{equation}\label{lol45}
\frac{Z_{n-1,\alpha}(\frac{x}{4(n-1)^{2}};r)}{\widehat Z_{n-1,\alpha+2}} = \bigO\big(e^{-(\frac{1}{4}-\epsilon)rx}\big), \qquad \mbox{ as } \qquad \frac{rx}{n}\to 0, \quad (r-1)x \to \infty.
\end{equation}
Note that because of the restriction $\frac{rx}{n}\to 0$, this asymptotic formula alone is not sufficient to estimate $I_{2}$. Nevertheless, we can derive the following inequality
\begin{equation}\label{lol44}
\begin{array}{r c l}
\hspace{-0.2cm}\displaystyle \frac{Z_{n-1,\alpha}(\frac{rx}{4(n-1)^{2}};r)}{\widehat Z_{n-1,\alpha+2}} & = & \displaystyle \frac{(n-1)!^{-1}}{\widehat Z_{n-1,\alpha+2}} \int_{\frac{rx}{4(n-1)^{2}}}^{\infty}\hspace{-0.2cm}...\int_{\frac{rx}{4(n-1)^{2}}}^{\infty} \Delta_{n-1}(\lambda)^{2}\prod_{i=1}^{n-1}\left(\lambda_{i}-\frac{x}{4(n-1)^{2}}\right)^{2}\lambda_{i}^{\alpha}e^{-(n-1)\lambda_{i}}d\lambda_{i} \\[0.2cm]
& \leq & \displaystyle \frac{1}{(n-1)!\widehat Z_{n-1,\alpha+2}} \int_{\frac{rx}{4(n-1)^{2}}}^{\infty}...\int_{\frac{rx}{4(n-1)^{2}}}^{\infty} \Delta_{n-1}(\lambda)^{2}\prod_{i=1}^{n-1}\lambda_{i}^{\alpha+2}e^{-(n-1)\lambda_{i}}d\lambda_{i} \\[0.4cm]
& = & \displaystyle \mathbb{P}_{n-1,\alpha+2}\left( \lambda_{\min} > \frac{rx}{4(n-1)^{2}} \right).
\end{array}
\end{equation}
It is well-known \cite{TraWidLUE} that the above quantity is bounded by $e^{- C x r}$ for sufficiently large $x$ as $n \to \infty$ (this is a large deviation principle), and where $C > 0$ is a constant. Combining \eqref{lol45} and \eqref{lol44}, we thus have
\begin{equation}
\begin{array}{r c l}
\displaystyle I_{2} & \leq & \displaystyle  \int_{M}^{\infty} x^{\alpha} e^{-\frac{x}{n}} e^{-Cxr} dx \leq e^{-\frac{C}{2}Mr}, \\
\end{array}
\end{equation}
if $M$ is chosen big enough. Therefore we obtain,
\begin{equation}\label{Prob almost final}
\lim_{n\to\infty} Q_{n,\alpha} (r) = \frac{1}{4^{\alpha+1}\Gamma(\alpha+1)\Gamma(\alpha+2)} \int_{0}^{M} x^{\alpha} e^{I(x;r)}dx + \bigO\big(e^{-\frac{C}{2} Mr}\big).
\end{equation}
Letting $M \to \infty$ in \eqref{Prob almost final} finishes the proof of Theorem \ref{thm-Q}.

\section{Appendix}
\subsection{Airy model RH problem}
\label{Subsection: Airy model RH problem}
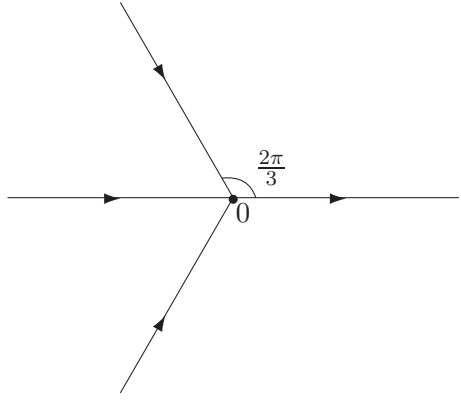
\begin{figure}[t]
    \begin{center}
    \setlength{\unitlength}{1truemm}
    \begin{picture}(100,55)(-5,10)
        \put(50,40){\line(1,0){30}}
        \put(50,40){\line(-1,0){30}}
        \put(50,39.8){\thicklines\circle*{1.2}}
        \put(50,40){\line(-0.5,0.866){15}}
        \put(50,40){\line(-0.5,-0.866){15}}
        \qbezier(53,40)(52,43)(48.5,42.598)
        \put(53,43){$\frac{2\pi}{3}$}
        \put(50.3,36.8){$0$}
        \put(65,39.9){\thicklines\vector(1,0){.0001}}
        \put(35,39.9){\thicklines\vector(1,0){.0001}}
        \put(41,55.588){\thicklines\vector(0.5,-0.866){.0001}}
        \put(41,24.412){\thicklines\vector(0.5,0.866){.0001}}
    \end{picture}
    \caption{\label{figAiry}The jump contour $\Sigma_{A}$ for $\tilde \Upsilon$.}
\end{center}
\end{figure}
\begin{itemize}
\item[(a)] $\widetilde{\Upsilon} (z) : \mathbb{C} \setminus \Sigma_{A} \rightarrow \mathbb{C}^{2 \times 2}$ is analytic where $\Sigma_{A}$ is shown in Figure \ref{figAiry}.
\item[(b)] $\widetilde{\Upsilon}$ has the jump relations
\begin{equation}\label{jumps P3}
\begin{array}{l l}
\widetilde{\Upsilon}_{+}(z) = \widetilde{\Upsilon}_{-}(z) \begin{pmatrix}
0 & 1 \\ -1 & 0
\end{pmatrix}, & \mbox{ on } \mathbb{R}^{-}, \\

\widetilde{\Upsilon}_{+}(z) = \widetilde{\Upsilon}_{-}(z) \begin{pmatrix}
 1 & 1 \\
 0 & 1
\end{pmatrix}, & \mbox{ on } \mathbb{R}^{+}, \\

\widetilde{\Upsilon}_{+}(z) = \widetilde{\Upsilon}_{-}(z) \begin{pmatrix}
 1 & 0  \\ 1 & 1
\end{pmatrix}, & \mbox{ on } e^{ \frac{2\pi}{3} i }  \mathbb{R}^{+} , \\

\widetilde{\Upsilon}_{+}(z) = \widetilde{\Upsilon}_{-}(z) \begin{pmatrix}
 1 & 0  \\ 1 & 1
\end{pmatrix}, & \mbox{ on }e^{ -\frac{2\pi}{3} i }\mathbb{R}^{+} . \\
\end{array}
\end{equation}
\item[(c)] As $z \to \infty$, we have
\begin{equation}\label{Asymptotics Airy}
\widetilde{\Upsilon}(z) = z^{-\frac{\sigma_{3}}{4}}N \left( I + \frac{1}{z^{3/2}} \widetilde{\Upsilon}_{1} + \bigO(z^{-3}) \right) e^{-\frac{2}{3}z^{3/2}\sigma_{3}},
\end{equation}
with 
\begin{equation}
\widetilde{\Upsilon}_{1} = \frac{1}{8} \begin{pmatrix}
\frac{1}{6} & i \\ i & -\frac{1}{6}
\end{pmatrix}.
\end{equation}
\item[(d)] As $z \to 0$, $\widetilde{\Upsilon}(z) = \bigO(1)$.
\end{itemize}
The following matrix-valued function solves the above Airy model RH problem (see \cite{ClaeysGravaMcLaughlin,Bleher}):
\begin{equation}
\widetilde{\Upsilon}(z) := M_{A} \times \left\{ \begin{array}{l l}
\begin{pmatrix}
\mbox{Ai}(z) & \mbox{Ai}(\omega^{2}z) \\
\mbox{Ai}^{\prime}(z) & \omega^{2}\mbox{Ai}^{\prime}(\omega^{2}z)
\end{pmatrix}e^{-\frac{\pi i}{6}\sigma_{3}}, & \mbox{for } 0 < \arg z < \frac{2\pi}{3}, \\
\begin{pmatrix}
\mbox{Ai}(z) & \mbox{Ai}(\omega^{2}z) \\
\mbox{Ai}^{\prime}(z) & \omega^{2}\mbox{Ai}^{\prime}(\omega^{2}z)
\end{pmatrix}e^{-\frac{\pi i}{6}\sigma_{3}}\begin{pmatrix}
1 & 0 \\ -1 & 1
\end{pmatrix}, & \mbox{for } \frac{2\pi}{3} < \arg z < \pi, \\
\begin{pmatrix}
\mbox{Ai}(z) & - \omega^{2}\mbox{Ai}(\omega z) \\
\mbox{Ai}^{\prime}(z) & -\mbox{Ai}^{\prime}(\omega z)
\end{pmatrix}e^{-\frac{\pi i}{6}\sigma_{3}}\begin{pmatrix}
1 & 0 \\ 1 & 1
\end{pmatrix}, & \mbox{for } -\pi < \arg z < -\frac{2\pi}{3}, \\
\begin{pmatrix}
\mbox{Ai}(z) & - \omega^{2}\mbox{Ai}(\omega z) \\
\mbox{Ai}^{\prime}(z) & -\mbox{Ai}^{\prime}(\omega z)
\end{pmatrix}e^{-\frac{\pi i}{6}\sigma_{3}}, & \mbox{for } -\frac{2\pi}{3} < \arg z < 0, \\
\end{array} \right.
\end{equation}
with $\omega = e^{\frac{2\pi i}{3}}$, Ai the Airy function and
\begin{equation}
M_{A} := \sqrt{2 \pi} e^{\frac{\pi i}{6}} \begin{pmatrix}
1 & 0 \\ 0 & -i
\end{pmatrix}.
\end{equation}

\subsection{Bessel model RH problem}
\label{Subsection: Bessel model RH problem}
This RH problem depends on a parameter $\alpha \in \mathbb{R}$.
\begin{itemize}
\item[(a)] $\Upsilon = \Upsilon^{(\alpha)} : \mathbb{C} \setminus \Sigma_{0,0}$ is analytic.
\item[(b)] $\Upsilon$ has the following jumps on $\Sigma_{0,0}\setminus \{0\}$:
\begin{align}
&\Upsilon_{+}(z) = \Upsilon_{-}(z) \begin{pmatrix}0&1\\-1&0\end{pmatrix}, &z\in \Sigma_{2},\\
&\Upsilon_{+}(z) = \Upsilon_{-}(z) \begin{pmatrix}
1 & 0 \\ e^{\pi i \alpha} & 1
\end{pmatrix}, & z\in \Sigma_{1},\\
&\Upsilon_{+}(z) = \Upsilon_{-}(z) \begin{pmatrix}
1 & 0 \\ e^{-\pi i \alpha} & 1
\end{pmatrix}, & z\in \Sigma_{3},
\end{align}
\item[(c)] As $z \rightarrow \infty$,
\begin{equation}
\Upsilon(z) = \left(I+\frac{1}{z}\Upsilon_{1} + \bigO(z^{-2})\right) z^{-\frac{1}{4}\sigma_3} N e^{z^{1/2} \sigma_{3}}.
\end{equation}
\item[(d)] As $z\to 0$,\\

(i) If $\alpha < 0$, $ \Upsilon(z) = \bigO \begin{pmatrix}
|z|^{\frac{\alpha}{2}} & |z|^{\frac{\alpha}{2}} \\
|z|^{\frac{\alpha}{2}} & |z|^{\frac{\alpha}{2}} \\
\end{pmatrix}$. \\

(ii) If $\alpha = 0$, $ \Upsilon(z) = \bigO \begin{pmatrix}
\log |z| & \log |z| \\
\log |z| & \log |z| \\
\end{pmatrix}$. \\

(iii) If $\alpha > 0$,

\begin{equation}
\Upsilon(z)=  \left\{ \begin{array}{l l}

\bigO \begin{pmatrix}
|z|^{\frac{\alpha}{2}} & |z|^{-\frac{\alpha}{2}} \\
|z|^{\frac{\alpha}{2}} & |z|^{-\frac{\alpha}{2}}
\end{pmatrix}, & \mbox{ for } -\frac{2\pi}{3} < \arg (z) < \frac{2\pi}{3}, \\

\bigO \begin{pmatrix}
|z|^{-\frac{\alpha}{2}} & |z|^{-\frac{\alpha}{2}} \\
|z|^{-\frac{\alpha}{2}} & |z|^{-\frac{\alpha}{2}}
\end{pmatrix}, & \mbox{ for } \arg(z) \in (-\pi,-\frac{2\pi}{3}) \cup (\frac{2\pi}{3},\pi). \\

\end{array} \right.
\end{equation}
\end{itemize}
The solution of this RH problem is explicitly given in terms of the Bessel functions (see \cite{KuiMcLVanVan} or \cite{AtkClaMez}), one has
\begin{equation}\label{lol4}
\Upsilon^{(\alpha)}(z) = \begin{pmatrix}
1 & 0 \\ \displaystyle \frac{i}{8}(4\alpha^2+3) & 1
\end{pmatrix} \pi^{\sigma_{3}/2} \begin{pmatrix}
\displaystyle I_{\alpha}(z^{1/2}) & \displaystyle \frac{i}{\pi}K_{\alpha}(z^{1/2}) \\[0.2cm]
\displaystyle \pi i z^{1/2}I_{\alpha}^{\prime}(z^{1/2}) & \displaystyle -z^{1/2}K_{\alpha}^{\prime}(z^{1/2})
\end{pmatrix} H_{0}(z),
\end{equation}
and
\begin{equation} \label{asymptotics_behaviour_of_Bessel_order1}
\begin{array}{r c l}
\displaystyle \left(\Upsilon_{1}\right)_{12} & = & \displaystyle \frac{i}{8} (2\alpha-1)(2\alpha+1).
\end{array}
\end{equation}

\section*{Acknowledgements}
The authors are grateful to Tom Claeys for useful discussions. MA and CC were supported by the European Research Council under the European Union's Seventh Framework Programme (FP/2007/2013)/ ERC Grant Agreement n.\, 307074, and CC was also supported by F.R.S.-F.N.R.S. SZ was funded by Nokia Technologies, Lockheed Martin and the University of Oxford through the quantum Optimisation and Machine Learning (QuOpaL) project.


\begin{thebibliography}{99}

\bibitem{AkeBaiFra} G. Akemann, J. Baik and P. D. Francesco, \textit{The Oxford Handbook of Random Matrix Theory}, Oxford Univ. Press, Oxford, 2011. 

\bibitem{AndGuiZei} G.W. Anderson, A. Guionnet, and O. Zeitouni, {\em An introduction to random matrices}, Cambridge studies in advanced mathematics \textbf{118} (2010). 

\bibitem{Atkin} M. Atkin, A Riemann-Hilbert problem for equations of Painlev\'{e} type in the one matrix model with semi-classical potential, to appear in \textit{Comm. Math. Phys.}

\bibitem{AtkClaMez} M. Atkin, T. Claeys, F. Mezzadri, Random matrix ensembles with singularities and a hierarchy of Painlev\'{e} III equations, \textit{Int. Math. Res. Notices} 2015, doi: 10.1093/imrn/rnv195.

\bibitem{BaiSil} Z. Bai, J.W. Silverstein. \textit{Spectral Analysis of Large Dimensional Random Matrices. Second Edition.} \textbf{20} (2010), Springer, New York.

\bibitem{Bleher} P. Bleher,
Lectures on random matrix models. The Riemann-Hilbert approach. Random matrices, random processes and integrable systems (2011), 251--349, 
{\em CRM Ser. Math. Phys.}, Springer, New York. 

\bibitem{CharlierDoeraene} C. Charlier and A. Doeraene, The generating function for the Bessel point process and a system of coupled Painlev\'{e} V equations, preprint.

\bibitem{ClaeysGravaMcLaughlin} T. Claeys, T. Grava, K. T-R McLaughlin, Asymptotics for the partition function in two-cut random matrix models, \textit{Comm. Math. Phys.} \textbf{339} (2015), no. 2, 513--587.

\bibitem{Deift} P. Deift, {\em Orthogonal Polynomials and Random Matrices: A Riemann-Hilbert Approach}, Amer. Math. Soc. \textbf{3} (2000).

\bibitem{DKMVZ2}
P. Deift, T. Kriecherbauer, K.T-R McLaughlin, S. Venakides, and X. Zhou, Uniform asymptotics for polynomials orthogonal with respect to varying exponential weights and applications to universality questions in random matrix theory, {\em Comm. Pure Appl. Math.} {\bf 52} (1999), 1335--1425. 
    
\bibitem{DKMVZ1}
P. Deift, T. Kriecherbauer, K.T-R McLaughlin, S. Venakides, and X. Zhou, Strong asymptotics of orthogonal polynomials with respect to exponential weights, {\em Comm. Pure Appl. Math.} {\bf 52} (1999), 1491--1552. 

\bibitem{DeiftZhou1992} P. Deift and X. Zhou, A steepest descent method for oscillatory Riemann-Hilbert problems. \textit{ Bull. Amer. Math. Soc. (N.S.)} \textbf{26} (1992), no. 1, 119--123.


\bibitem{FokasItsKapaevNovokshenov} A.S. Fokas, A.R. Its, A.A. Kapaev and V.Y. Novokshenov, \textit{Painlev\'{e} Transcendents: The Riemann-Hilbert Approach}, American Mathematical Soc. \textbf{128}, 2006.

\bibitem{FokasItsKitaev}
A.S. Fokas, A.R. Its, and A.V. Kitaev, The isomonodromy approach to matrix models in 2D quantum gravity, \emph{Comm. Math. Phys.} \textbf{147} (1992), no. 2, 395--430.

\bibitem{ForHug} P. J. Forrester, T. D. Hughes, Complex Wishart matrices and conductance in mesoscopic systems: exact results, \textit{J. Math. Phys.} \textbf{35} (1994), no. 12, 6736--6747. 

\bibitem{ForWit} P. J. Forrester, N. S. Witte, The distribution of the first eigenvalue spacing at the hard edge of the Laguerre unitary ensemble,  \textit{Kyushu J. Math.} \textbf{61} (2007), no. 2, 457--526. 

\bibitem{HasMcL} S. P. Hastings, J. B. McLeod, A boundary value problem associated with the second Painlev\'{e} transcendent and the Korteweg-de Vries equation. \textit{ Arch. Rational Mech. Anal.} \textbf{73} (1980), no. 1, 31--51.

\bibitem{KeaSna} J.P. Keating and N.C. Snaith, Random matrix theory and $\zeta(1/2+it)$, {\em Comm. Math. Phys.} {\bf 214} (2001), 57--89. 


\bibitem{KuiMcLVanVan} A.B.J. Kuijlaars, K. T.-R. McLaughlin, W. Van Assche, M. Vanlessen, The Riemann-Hilbert approach to strong asymptotics for orthogonal polynomials on $[-1,1]$. \textit{ Adv. Math.} \textbf{188} (2004), no. 2, 337--398. 

\bibitem{MarPas} V. A. Marchenko, L. A. Pastur, Distribution of eigenvalues for some sets of random matrices, \textit{Mat. Sb. (N.S.)} \textbf{72} (1967), 507--536.

\bibitem{Mehta} Mehta M.L., Random matrices, Third Edition, \textit{Pure and Applied Mathematics Series} \textbf{142}, Elsevier Academic Press, 2004.

\bibitem{NIST} F.W.J. Olver, A.B. Olde Daalhuis, D.W. Lozier, B.I. Schneider, R.F. Boisvert, C.W. Clark, B.R. Miller, and B.V. Saunders, NIST Digital Library of Mathematical Functions. http://dlmf.nist.gov/, Release 1.0.13 of 2016-09-16.

\bibitem{PerSch} A. Perret, G. Schehr, Near-extreme eigenvalues and the first gap of Hermitian random matrices, \textit{J. Stat. Phys.} \textbf{156} (2014), no. 5, 843--876.

\bibitem{PerSch2} A. Perret, G. Schehr, The density of eigenvalues seen from the soft edge of random matrices in the Gaussian $\beta$-ensembles, \textit{ Acta Phys. Polon. B} \textbf{46} (2015), no. 9, 1693--1707.

\bibitem{Szego OP} G. Szeg\H{o}, \textit{Orthogonal polynomials}, AMS Colloquium Publ. \textbf{23} (1959), New York: AMS.

\bibitem{SaTo} E. B. Saff and V. Totik, {\em Logarithmic Potentials with External Fields}, Springer-Verlag (1997).

\bibitem{Tao} T. Tao, \textit{Topics in random matrix theory}, Amer. Math. Soc. \textbf{132} (2012).

\bibitem{TraWidGUE} C. A. Tracy, H. Widom. Level-spacing distributions and the Airy kernel. \textit{Phys. Lett. B} \textbf{305} (1993), no. 1-2, 115--118.

\bibitem{TraWidLUE} C. A. Tracy, H. Widom. Level spacing distributions and the Bessel kernel.  \textit{Comm. Math. Phys.} \textbf{161} (1994), no. 2, 289--309.

\bibitem{Wishart} J. Wishart, The generalized product moment distribution in samples from a normal multivariate population, \textit{Biometrika A} \textbf{20} (1928), 32--52.

\bibitem{WitBorFor} N.S. Witte, F. Bornemann, P.J. Forrester, Joint distribution of the first and second eigenvalues at the soft edge of unitary ensembles, \textit{Nonlinearity} \textbf{26} (2013), no. 6, 1799--1822.

\bibitem{XuZhao} S.-X. Xu, Y.-Q. Zhao, Critical edge behavior in the modified Jacobi ensemble and the Painlev\'{e} V transcendents, \textit{J. Math. Phys.} \textbf{54} (2013), no. 8, 083304, 29 pp.

\bibitem{XuZhao2} S.-X. Xu, Y.-Q. Zhao, Critical edge behavior in the modified Jacobi ensemble and Painlev\'{e} equations, \textit{ Nonlinearity} \textbf{28} (2015), no. 6, 1633–1674.



\end{thebibliography}
\end{document}